\documentclass[letterpaper,twocolumn,10pt]{article}
\usepackage{usenix2019_v3}

\usepackage{amsthm,amsmath,amssymb,amsfonts}
\usepackage{graphicx}
\usepackage{textcomp}
\usepackage{xcolor}
\usepackage[font=small,labelfont=bf]{caption}
\usepackage{verbatim}
\usepackage{subfig}
\usepackage{multirow}
\usepackage{dsfont}
\usepackage{bbding}

\usepackage{algorithm}
\usepackage{algorithmic}

\renewenvironment{thebibliography}[1]{
	\begin{oldthebibliography}{#1}
		\setlength{\itemsep}{-0.033em}
		\setlength{\parskip}{-0.033em}
	}
	{
	\end{oldthebibliography}
}

\theoremstyle{plain}
\newtheorem{theorem}{Theorem}
\newtheorem{lemma}{Lemma}
\newtheorem{definition}{Definition}

\usepackage{fancyhdr}
\fancypagestyle{firststyle}
{
	\fancyhf{}
	
	\fancyhead[C]{\emph{To appear in the 29th Usenix Security Symposium, August 2020, Boston, MA}}
}

\begin{document}
	
\thispagestyle{firststyle}

\title{\Large \bf PCKV: Locally Differentially Private  Correlated Key-Value \\Data Collection with Optimized Utility}

\author{{\rm Xiaolan Gu}\\ University of Arizona\\{\rm xiaolang@email.arizona.edu} \and {\rm Ming Li}\\ University of Arizona\\{\rm lim@email.arizona.edu} \and {\rm Yueqiang Cheng{\textsuperscript{{\Envelope}}}}\\ Baidu X-Lab \\{\rm chengyueqiang@baidu.com} \and {\rm Li Xiong}\\ Emory University\\{\rm lxiong@emory.edu} \and {\rm Yang Cao}\\ Kyoto University\\{\rm yang@i.kyoto-u.ac.jp}}

\maketitle

\begin{abstract}
Data collection under local differential privacy (LDP) has been mostly studied for homogeneous data. Real-world   applications often involve a mixture of different data types such as key-value pairs, where the  frequency of keys and mean of values under each key must be estimated simultaneously. For key-value data collection with LDP,  it is challenging to achieve a good utility-privacy tradeoff since the data contains two dimensions and a user may possess multiple key-value pairs. There is also an inherent correlation between key and values which if not harnessed, will lead to poor utility.  In this paper, we propose a locally differentially private key-value  data collection   framework that utilizes correlated perturbations  to enhance  utility. We instantiate our framework by two protocols PCKV-UE (based on Unary Encoding) and PCKV-GRR (based on Generalized Randomized Response),  where we design an  advanced Padding-and-Sampling mechanism and an improved mean estimator which is non-interactive. Due to  our correlated key and value perturbation mechanisms, the composed privacy budget is  shown to be less  than that of independent perturbation of key and value, which enables us to further optimize the perturbation parameters via budget allocation.  Experimental results on both synthetic and real-world datasets show that our proposed protocols achieve better utility for both frequency and mean estimations under the same LDP guarantees than state-of-the-art mechanisms.
\end{abstract}

\section{Introduction}

Differential Privacy (DP) \cite{dwork2006differential,dwork2006calibrating} has become the \emph{de facto} standard for private data release. It provides provable privacy protection, regardless of the adversary's background knowledge and computational power \cite{chen2016private}.  In recent years, Local Differential Privacy (LDP) has been proposed to protect privacy at the data collection stage, in contrast to DP in the centralized setting which protects data after it is collected and stored by a server. In the local setting, the server is assumed to be untrusted, and each user independently perturbs her raw data  using a privacy-preserving mechanism that satisfies LDP. Then, the server collects the perturbed data from all users to perform data analytics or answer queries from users or third parties. The local setting has been widely adopted in practice. For example, Google's RAPPOR \cite{erlingsson2014rappor}  has been employed in Chrome to collect web browsing behavior with LDP guarantees; Apple is also using LDP-based mechanisms to identify popular emojis, popular health data types, and media playback preference in Safari \cite{apple2017learning}.

\begin{figure}[!t]
    \centering
    \includegraphics[width=3.4in]{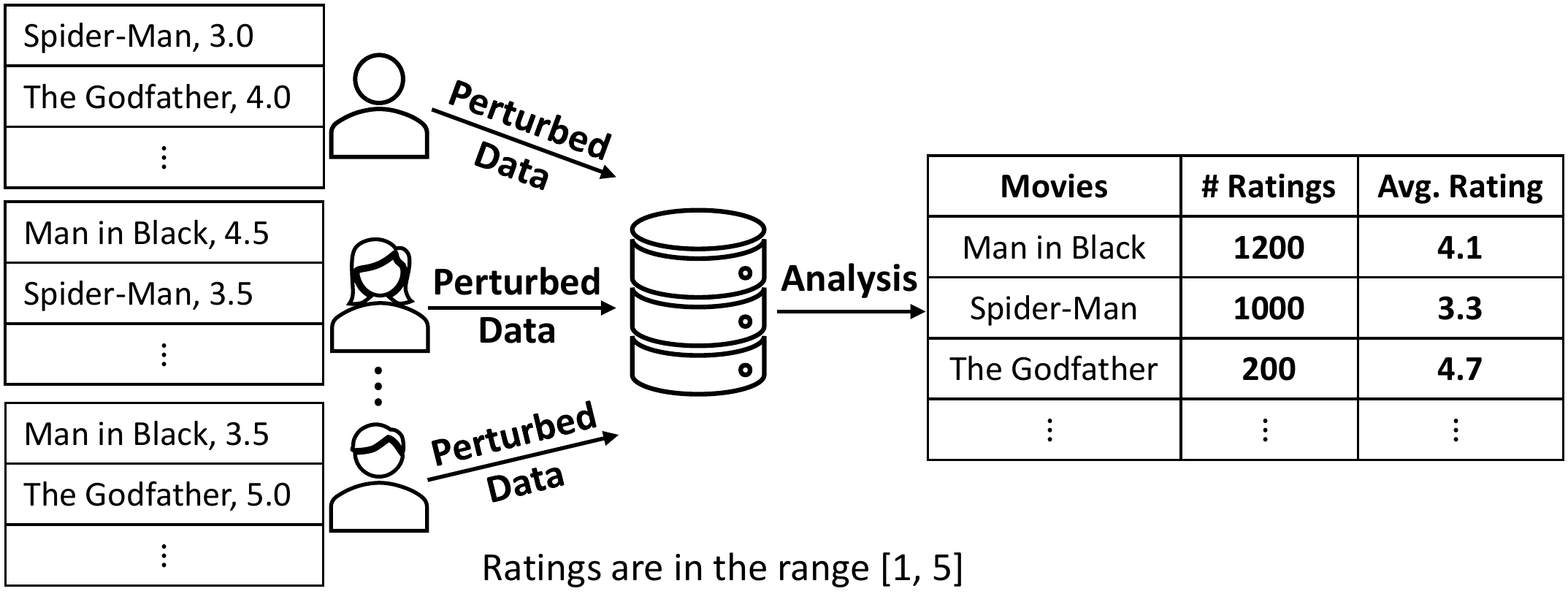}
    \vspace{-8mm}
    \caption{A motivating example (movie rating system).}
    \label{fig:model}
    \vspace{-8mm}
\end{figure}

Early works under LDP mainly focused on simple statistical queries such as frequency/histogram estimation on categorical data \cite{wang2017locally} and mean estimation of numerical data \cite{duchi2018minimax,nguyen2016collecting,ding2017collecting}. Later works studied more complex queries or structured data, such as frequent item/itemset mining of itemset data \cite{qin2016heavy,wang2018locally}, computing mean value over a single numeric attribute of multidimensional data \cite{wang2019collecting,zhang2018calm,ren2018textsf}, and generating synthetic social graphs from graph data \cite{qin2017generating}. However, few of them studied the hybrid/heterogeneous data types or queries (e.g., both categorical and numerical data). Key-value data is one such example, which is widely encountered in practice. 

As a motivating example, consider a movie rating system (shown in Figure \ref{fig:model}), each user possesses multiple records of movies (the keys) and their corresponding ratings (the values), that is, a set of key-value pairs. The data collector (the server) can aggregate the rating records from all users and analyze the statistical property of a certain movie, such as the ratio of people who watched this movie (frequency) and the average rating (value mean). Then, the server (or a third party) can provide recommendations by choosing movies with both high frequencies and large value means. 

The main challenges to achieve high utility for   key-value data collection under LDP are two-fold: multiple key-value pairs possessed by each user and the inherent correlation between the key and value. For the former, if all the key-value pairs of a user are reported to the server,  each pair will split the limited privacy budget $\epsilon$ (the larger $\epsilon$ is, the more leakage is allowed), which requires more noise/perturbation for each pair. For the latter, correlation means reporting the value of a key also discloses information about the presence of that key. If the key and value are independently perturbed each under $\epsilon$-LDP,   overall it satisfies $2\epsilon$-LDP according to sequential composition, which means more perturbation is needed for both key and value to satisfy $\epsilon$-LDP overall. Intuitively, jointly perturbing key and value by exploiting such correlation may lead to less overall leakage; however, it is non-trivial to design such a mechanism   that substantially improves the budget composition. 
 
Recently,  Ye et al. \cite{ye2019privkv} are the first to propose PrivKVM to estimate the frequency and mean of key-value data. Because of key-value correlation, they adopt an interactive protocol with multiple rounds used to iteratively improve the estimation of a key’s mean value.  The mean estimation in PrivKVM is shown to be unbiased when the number of iterations is large enough. However, it has three major limitations. First,  multiple rounds will enlarge the variance of mean estimation (as the privacy budget is split in each iteration) and reduce the practicality (since users need to be online). Second, they use a sampling protocol that samples an index from the domain of all keys to address the first challenge,  which does not work well for a large key domain (explained in Sec. \ref{sec:PrivKVM}). Third, although their mechanism considers the correlation between key and value, it  does not lead to an improved budget composition for LDP (discussed in Sec. \ref{sec:perturbation}).

In this paper, we propose a novel framework for Locally Differentially Private Correlated Key-Value (PCKV) data collection  with a better utility-privacy tradeoff. It enhances PrivKVM in four aspects, where the first three address the limitations of PrivKVM, and the last one further improves the utility based on  optimized budget allocation.

First,   we propose an improved mean estimator which only needs a single-round. We divide   the calibrated sum of values of a certain key by the calibrated  frequency of that key (whose expectation is the true frequency of  keys), unlike PrivKVM which uses  uncalibrated versions of both (value sum and frequency) that is skewed by inputs from the fake keys and their values. To fill the values of fake keys,  we only need to  randomly generate values with zero mean (which do not change the expectation of estimated value sum), eliminating the need to iteratively estimate the mean for fake value generation. Although the division of two unbiased estimators is not unbiased in general, we show that it is a consistent estimator (i.e., the bias converges to 0 when the number of users increases).  We also propose an improved estimator to correct the outliers when estimation error is large under a small $\epsilon$.

Second, we adapt an advanced sampling protocol called Padding-and-Sampling \cite{wang2018locally} (originally used in itemset data) to sample one key-value pair from the local pairs that are possessed by the user to make sure most of sampled data are useful. Such an advanced sampling protocol can enhance utility, especially for a large domain size.


Third, as a byproduct of uniformly random fake value generation (when a non-possessed key is reported as possessed), we show that the proposed correlated perturbation strategy consumes less privacy budget overall than the budget summation of key and value perturbations, by deriving a tighter bound of the composed privacy budget (Theorem \ref{thm:LDP_UE} and Theorem \ref{thm:LDP_GRR}). It can provide a better utility-privacy tradeoff than using the basic sequential composition of LDP  which assumes independent mechanisms. Note that PrivKVM directly uses sequential composition for privacy analysis.

Fourth, since the Mean Square Error (MSE) of frequency and mean estimations in our scheme can be theoretically analyzed (in Theorem \ref{thm:estimation}) with respect to the two privacy budgets of key and value perturbations, it is possible to find the optimized budget allocation with minimum MSE under a given privacy constraint (budget). However, the MSEs depend  on the true frequency and value mean that are unknown in practice. Thus, we  derive near-optimal privacy budget allocation and perturbation parameters in closed-form (Lemma \ref{lem:budget_UE} and Lemma \ref{lem:budget_GRR}) by minimizing an approximate upper bound of the MSE.  Our near-optimal allocation is shown (in both theoretical and empirical) to outperform the naive budget allocation with an equal split. 

Main contributions are summarized as follows:

(1) We propose the PCKV framework with two  mechanisms PCKV-UE and PCKV-GRR under two baseline perturbation protocols: Unary Encoding (UE) and Generalized Randomized Response (GRR). Our scheme is non-interactive (compared with PrivKVM) as the mean of values is estimated in one round. We theoretically analyze the expectation and MSE and show its asymptotic unbiasedness. 

(2) We adapt the Padding-and-Sampling protocol \cite{wang2018locally} for key-value data, which handles large domain better than the sampling protocol used in PrivKVM.

(3) We show the budget composition of our correlated perturbation mechanism, which has a tighter bound than using the sequential composition of LDP.

(4) We propose a near-optimal budget allocation approach with closed-form solutions for PCKV-UE and PCKV-GRR under the tight budget composition. The utility-privacy tradeoff of our scheme is improved by both the tight budget composition and the optimized budget allocation.

(5) We evaluate our scheme using both synthetic  and real-world datasets, which is shown to have higher utility (i.e., less MSE) than existing schemes. Results also validate the correctness of our theoretical analysis  and the improvements of the tight budget composition and optimized budget allocation.

\section{Related Work}

The main task of local differential privacy techniques is to analyze some statistic information from the data that has been perturbed by users. Erlingsson et al. \cite{erlingsson2014rappor} developed RAPPOR satisfying LDP for Chrome to collect URL click counts. It is based on the ideas of Randomized Response \cite{warner1965randomized}, which is a technique for collecting statistics on sensitive queries when a respondent wants to retain confidentiality. In the basic RAPPOR, they adopt unary encoding to obtain better performance of frequency estimation. Wang et al. \cite{wang2017locally} optimized the parameters of basic RAPPOR by minimizing the variance of frequency estimation. There are a lot of works that focus on complex data types and complex analysis tasks under LDP. Bassily and Smith \cite{bassily2015local} proposed an asymptotically optimal solution for building succinct histograms over a large categorical domain under LDP.  Qin et al. \cite{qin2016heavy} proposed a two-phase work named LDPMiner to achieve the heavy hitter estimation (items that are frequently possessed by users) over the set-valued data with LDP, where each user can have any subset of an item domain with different length. Based on the work of LDPMiner, Wang et al. \cite{wang2018locally} studied the same problem and proposed a more efficient framework to estimate not only the frequent items but also the frequent itemsets. %

To the best of our knowledge, there are only two works on key-value data collection under LDP. Ye et al. \cite{ye2019privkv} are the first to propose PrivKV,
PrivKVM, and PrivKVM$^{+}$, where PrivKVM iteratively estimates the mean to guarantee the unbiasedness. PrivKV can be regarded as PrivKVM with only one iteration. The advanced version PrivKVM$^{+}$ selects a proper number of iterations to balance the unbiasedness and communication cost. Sun et al. \cite{sun2019conditional} proposed another estimator for frequency and mean under the framework of PrivKV and several  mechanisms to accomplish the same task. They also introduced conditional analysis (or the marginal statistics) of key-value data for other complex analysis tasks in machine learning. However, both of them use the naive sampling protocol and neither of them analyzes the tighter budget composition caused by the correlation between perturbations nor considers the optimized budget allocation.

\section{Preliminaries}

\subsection{Local Differential Privacy}
In the centralized setting of differential privacy, the data aggregator (server) is assumed to be trusted who possesses all users' data and perturbs the query answers. However, this assumption does not always hold in practice and may not be convincing enough to the users. 
In the local setting, each user perturbs her input $x$ using a mechanism $\mathcal{M}$ and uploads $y=\mathcal{M}(x)$ to the server for data analysis, where the server can be untrusted because only the user possesses the raw data of herself; thus the server has no direct access to the raw data. 

\begin{definition}[\textbf{Local Differential Privacy (LDP) \cite{duchi2013local}}]
	For a given $\epsilon\in\mathbb{R}^{+}$, a randomized mechanism $\mathcal{M}$ satisfies $\epsilon$-LDP if and only if for any pair of inputs $x,x^{\prime}$, and any output $y$, the probability ratio of outputting the same $y$ should be bounded 
	\begin{align}
	\label{equ:def_LDP}
	\frac{\Pr(\mathcal{M}(x)=y)}{\Pr(\mathcal{M}(x^{\prime})=y)} \leqslant e^{\epsilon}
	\end{align}
\end{definition}

Intuitively, given an output $y$ of a mechanism, an adversary cannot infer with high confidence (controlled by $\epsilon$) whether the input is $x$ or $x^{\prime}$, which provides plausible deniability for individuals involved in the sensitive data. Here, $\epsilon$ is a parameter called  \emph{privacy budget} that controls the strength of privacy protection. A smaller $\epsilon$ indicates stronger privacy protection because the adversary has lower confidence when trying to distinguish any pair of inputs $x,x^{\prime}$. A very good property of LDP is sequential composition, which guarantees the overall privacy for a sequence of mechanisms that satisfy LDP. 
\begin{theorem}[\textbf{Sequential Composition of LDP \cite{mcsherry2009privacy}}]
    \label{thm:seq_compo}
	If a randomized mechanism $\mathcal{M}_i: \mathcal{D}\rightarrow\mathcal{R}_i$ satisfies $\epsilon_i$-LDP for $i=1,2,\cdots,k$, then their sequential composition $\mathcal{M}: \mathcal{D}\rightarrow\mathcal{R}_1\times\mathcal{R}_2\times\cdots\times\mathcal{R}_k$ defined by $\mathcal{M}=(\mathcal{M}_1,\mathcal{M}_2,\cdots,\mathcal{M}_k)$ satisfies $(\sum_{i=1}^k\epsilon_i)$-LDP.
\end{theorem}
According to sequential composition, a given privacy budget for a computation task can be split into multiple portions, where each portion corresponds to the budget for a sub-task. 

\subsection{Mechanisms under LDP}
\label{sec:LDP mechanism}

\textbf{Randomized Response.} Randomized Response (RR) \cite{warner1965randomized} is a technique developed for the interviewees in a survey to return a randomized answer to a sensitive question so that the interviewees can enjoy plausible deniability. Specifically, each interviewee gives a genuine answer with probability $p$ or gives the opposite answer with probability $q=1-p$.  In order to satisfy $\epsilon$-LDP, the probability is selected as $p=\frac{e^\epsilon}{e^\epsilon+1}$. RR only works for binary data, but it can be extended to apply for the general category set $\{1,2,\cdots,d\}$ by Generalized Randomized Response (GRR) or  Unary Encoding (UE). 

\textbf{Generalized Randomized Response.} 
The perturbation function in Generalized Randomized Response (GRR) \cite{wang2017locally} is 
\begin{align*}
\Pr(\mathcal{M}(x)=y)=
\begin{cases}
p=\frac{e^\epsilon}{e^\epsilon+d-1}, & \text{if } y=x\\
q=\frac{1-p}{d-1}, & \text{if } y\neq x
\end{cases}
\end{align*}
where $x,y\in\{1,2,\cdots,d\}$ and the values of $p$ and $q$ guarantee $\epsilon$-LDP of the perturbation (because $\frac{p}{q}=e^{\epsilon}$).

\textbf{Unary Encoding.} 
The Unary Encoding (UE) \cite{wang2017locally} converts an input $x=i$ into a bit vector $\mathbf{x}=[0,\cdots,0,1,0,\cdots,0]$ with length $d$, where only the $i$-th position is 1 and other positions are 0s. Then each user perturbs each bit of  $\mathbf{x}$ independently with the following probabilities ($q\leqslant0.5\leqslant p$)
\begin{align*}
\Pr(\mathbf{y}[k]=1)=
\begin{cases}
p, & \text{if } \mathbf{x}[k]=1\\
q, & \text{if } \mathbf{x}[k]=0
\end{cases}\quad
(\forall k=1,2,\cdots,d)
\end{align*}
where $\mathbf{y}$ is the output vector with the same size as vector $\mathbf{x}$. It was shown in \cite{wang2017locally} that this mechanism satisfies LDP with $\epsilon=\ln\frac{p(1-q)}{(1-p)q}$. The selection of $p$ and $q$ under a given privacy budget $\epsilon$ varies for different mechanisms. For example, the basic RAPPOR \cite{erlingsson2014rappor} assigns $p=\frac{e^{\epsilon/2}}{e^{\epsilon/2}+1}$ and $q=1-p$, while the Optimized Unary Encoding (OUE) \cite{wang2017locally} assigns $p=\frac{1}{2}$ and $q=\frac{1}{e^{\epsilon}+1}$, which is obtained by minimizing the approximate variance of frequency estimation.

\textbf{Frequency Estimation for GRR, RAPPOR and OUE.}
After receiving the perturbed data from all users (with size $n$), the server can compute the observed proportion of users who possess the $i$-th item (or $i$-th bit), denoted by $f_i$. Since the perturbation is biased for different items (or bit-0 and bit-1), the server needs to estimate the observed frequency by an unbiased estimator $\hat{f}_i=\frac{f_i-q}{p-q}$, whose Mean Square Error (MSE) equals to its variance \cite{wang2017locally}
\begin{align*}
\text{MSE}_{\hat{f}_i}=\text{Var}[\hat{f}_i]=\frac{q(1-q)}{n(p-q)^2}+\frac{f_i^{*}(1-p-q)}{n(p-q)}
\end{align*}
where $f_i^{*}$ is the ground truth of the frequency for item $i$.

\section{Key-Value Data Collection under LDP}

\subsection{Problem Statement}


\textbf{System Model.}
Our system model (shown in Figure \ref{fig:model}) involves one data server and a set of users $\mathcal{U}$ with size $|\mathcal{U}|=n$. Each user possesses one or multiple key-value pairs $\langle k,v\rangle$, where $k\in\mathcal{K}$ (the domain of key) and $v\in\mathcal{V}$ (the domain of value). We assume the domain size of key is $d$, i.e., $\mathcal{K}=\{1,2,\cdots,d\}$, and domain of value is $\mathcal{V}=[-1,1]$ (any bounded value space can be linearly transformed into this domain). The set of key-value pairs possessed by a user is denoted as $\mathcal{S}$ (or $\mathcal{S}_u$ for a specific user $u\in\mathcal{U}$). After collecting the perturbed data from all users, the server needs to estimate the frequency (the proportion of users who possess a certain key) and the value mean (the averaged value of a certain key from the users who possess such key), i.e.,
\begin{align*}
    f_k^{*}=\frac{\sum_{u\in\mathcal{U}}\mathds{1}_{\mathcal{S}_u}(\langle k,\cdot\rangle)}{n},\quad
    m_k^{*} = \frac{\sum_{u\in\mathcal{U},\langle k,v\rangle\in\mathcal{S}_u}v}{n\cdot f_k^{*}}
\end{align*}
where $\mathds{1}_{\mathcal{S}_u}(\langle k,\cdot\rangle)$ is 1 when $\langle k,\cdot\rangle\in\mathcal{S}_u$ and is 0 otherwise.

\textbf{Threat Model.} We assume the server is untrusted and each user only trusts herself because the privacy leakage can be caused by either unauthorized data sharing or breach due to hacking activities. Therefore, the adversary is assumed to have access to the output data of all users and know the perturbation mechanism adopted by the users. Note that we assume all users are honest in following the perturbation mechanism, thus we do not consider the case that some users maliciously upload bad data to fool the server.

\textbf{Objectives and Challenges.} Our goal is to estimate frequency and mean with high accuracy (i.e., small Mean Square Error) under the required privacy constraint (i.e., satisfying $\epsilon$-LDP). However, the task is not trivial for key-value data due to the following challenges: \textbf{(1)} Considering each user can possess multiple key-value pairs (the number of pairs can be different for users), if each user uploads multiple pairs, then each pair needs to consume budget, leading to a smaller budget and larger noise in each pair. On the other hand, if simply sampling an index $j$ from the domain and uploading the key-value pair regarding the $j$-th key (which is used in PrivKVM \cite{ye2019privkv}), we cannot make full use of the original pairs. Therefore, an elaborately designed sampling protocol is necessary in order to estimate the frequency and mean with high accuracy. \textbf{(2)} Due to the correlation between key and value in a key-value pair, the perturbation of key and value should be correlated. If a user reports a key that does not exist in her local data, she has to generate a fake value to guarantee the indistinguishability; however, how to generate the fake value without any prior knowledge and how to eliminate the influence of fake values on the mean estimation are challenging tasks. \textbf{(3)} Considering the key and value are perturbed in a correlated manner, the overall perturbation mechanism   may not leak as much information as two independent perturbations do (by sequential composition). Therefore, precisely quantifying the actually consumed privacy budget can improve the privacy-utility tradeoff of the overall key-value perturbation.

\subsection{PrivKVM}
\label{sec:PrivKVM}
To the best of our knowledge, PrivKVM \cite{ye2019privkv} is the only published work on key-value data collection in the LDP setting (note that another existing work \cite{sun2019conditional} is a preprint). It utilizes one iteration for frequency estimation and multiple iterations to approximately approach the unbiased mean estimation. We briefly describe it as follows. Assume the total privacy budget is $\epsilon$, and the number of iterations is $c$. In the \emph{first iteration}, each user randomly samples an index $j$ from the key domain $\mathcal{K}$ with uniform distribution (note that $j$ does not contain any private information). If the user processes key $k=j$ with value $v$, then she perturbs the key-value pair $\langle 1,v\rangle$; if not, the user perturbs the key-value pair $\langle 0,\tilde{v}\rangle$, where $\tilde{v}$ is initialized as 0 in the first iteration. In both cases, the input is perturbed with key-budget $\frac{\epsilon}{2}$ and value-budget $\frac{\epsilon}{2c}$. Then, each user uploads the index $j$ and one perturbed key-value pair $\langle 0,\cdot\rangle$ or $\langle 1,\cdot\rangle$ to the server and the server can compute the estimated frequency $f_k$ and mean $m_k~(k\in\mathcal{K})$ after collecting the perturbed data from all users, where the counts of output values will be corrected before estimation when outliers occur. In the \emph{remaining iterations}, each user perturbs her data with a similar way but $\tilde{v}=m_k$ (the estimated mean of the previous round) and the budget for key perturbation is 0. Then, the server updates the mean $m_k$ in the current iteration. By multiple rounds of interaction between users and the server, the mean estimation is approximately unbiased, and the sequential composition guarantees LDP with privacy budget $\frac{\epsilon}{2}+\frac{\epsilon}{2c}\cdot c=\epsilon$. 

There are three limitations of PrivKVM. 

(1) To achieve approximate unbiasedness, PrivKVM needs to run multiple rounds. This requires all users online during all rounds, which is impractical in many application scenarios. Also, the multiple iterations only guarantee the convergence of expectation of mean estimation (i.e., the bias theoretically approaches zero when $c\rightarrow\infty$), but the variance of mean estimation will be very large for a large $c$ because the budget $\frac{\epsilon}{2c}$ (for value perturbation in each round) is very small. Note that the estimation error depends on both bias and variance.

(2) The sampling protocol in PrivKVM may not work well for a large domain. When the domain size $d=|\mathcal{K}|$ is very large (such as millions) and each user only has a relatively small number of key-value pairs (such as less than 10), uniformly sampling an index from the large key domain $\mathcal{K}$ makes users rarely upload the information of the keys that they possess, resulting in a large variance of frequency and mean estimations. Also, when the number of users $n$ is not very large compared with domain size (such as $n<2d$), some keys may not be sampled, then the mean estimation does not work for such keys because of no samples.

(3) Although PrivKVM considers the correlation between key and value, it does not lead to an improved budget composition for LDP, which will be discussed in Sec. \ref{sec:perturbation}.

\section{Proposed Framework and Mechanisms}
\begin{figure}
    \centering
    \includegraphics[width=3.4in]{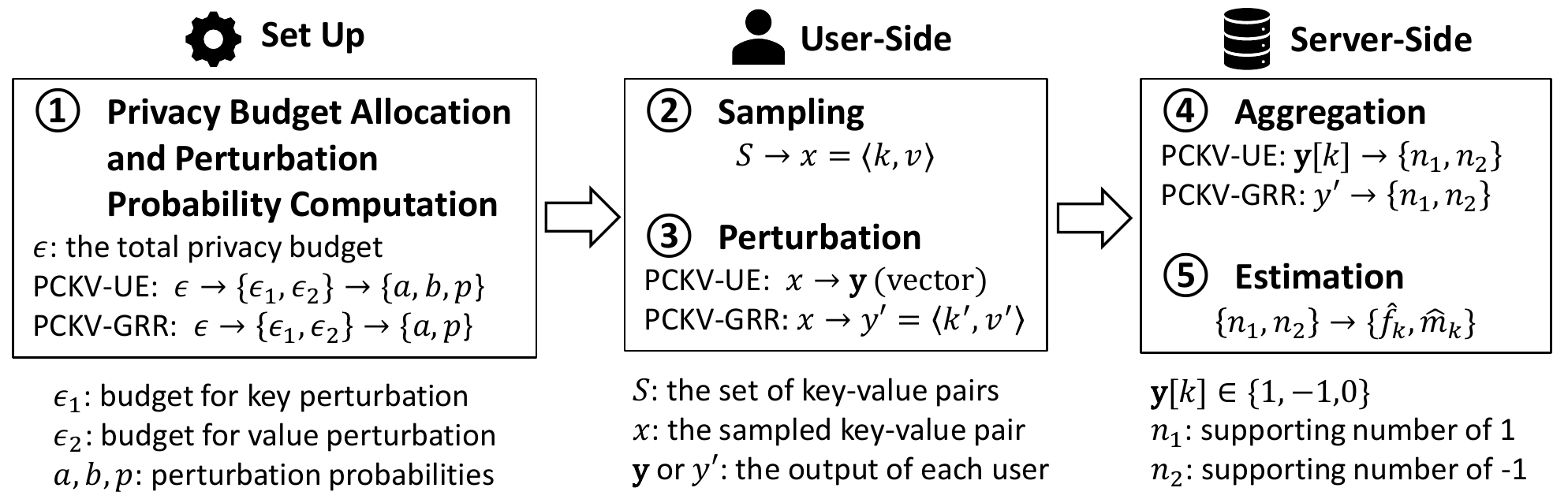}
    \vspace{-7mm}
    \caption{The overview of our PCKV framework.}
    \vspace{-6mm}
    \label{fig:Diagram}
\end{figure}

The overview of our PCKV framework is shown in Figure \ref{fig:Diagram}, where two specific mechanisms are included. The first one is PCKV-UE, which outputs a bit vector, and the second one is PCKV-GRR, which outputs a key-value pair. Note that the two mechanisms have similar ideas but steps \textcircled{\footnotesize 1}\textcircled{\footnotesize 3}\textcircled{\footnotesize 4} are slightly different. In step \textcircled{\footnotesize 1}, the system sets up some environment parameters (such as the total budget $\epsilon$ and domain size $d$),  which can be used to allocate the privacy budget for key and value perturbations and compute the perturbation probabilities in mechanisms, where the optimized privacy budget allocation is discussed in Sec. \ref{sec:allocation}. In step \textcircled{\footnotesize 2} and step  \textcircled{\footnotesize 3}, each user samples one key-value pair from her local data and privately perturbs it, where the sampling protocol is discussed in Sec. \ref{sec:sampling} and the perturbation mechanisms (PCKV-UE and PCKV-GRR) are proposed in Sec. \ref{sec:perturbation}. The perturbation of  value depends on the perturbation of key, which is utilized to improve the privacy budget composition. In step \textcircled{\footnotesize 4} and step \textcircled{\footnotesize 5}, the server aggregates the perturbed data from all users and estimates the frequency and mean, shown in Sec. \ref{sec:calibratoin}.

\subsection{Sampling Protocol}
\label{sec:sampling}
This subsection corresponds to step \textcircled{\footnotesize 2} in Figure \ref{fig:Diagram}. Considering each user may possess multiple key-value pairs, if the user perturbs and uploads all pairs, then each pair would consume the budget and the noise added in each pair becomes too large. Therefore, a promising solution is to upload the perturbed data of one pair (by sampling) to the server, which can avoid budget splitting.  As analyzed in Sec.
\ref{sec:PrivKVM}, the sampling protocol used in PrivKVM does not work well for a large domain. In this paper, we use an advanced protocol called Padding-and-Sampling \cite{wang2018locally} to improve the performance.

The  Padding-and-Sampling protocol \cite{wang2018locally} is originally used for itemset data, where each user samples one item from possessed items rather than sampling an index from the domain of all items. To make the sampling rate the same for all users, each user first pads her items into a uniform length $\ell$ by some dummy items from a domain of size $\ell$. Although there may still exist unsampled items, this case occurs only for infrequent items, thus the useful information of frequent items still can be reported with high probability. 

\setlength{\textfloatsep}{5pt}
\begin{algorithm}[!t]
	\caption{Padding-and-Sampling for Key-Value Data}
	\footnotesize
	\begin{algorithmic}[1]
		\REQUIRE The set of key-value pairs $\mathcal{S}$, padding length $\ell$
		\ENSURE Sampled key-value pair $\langle k,v\rangle$, where $k\in\mathcal{K}^{\prime}$ and $v\in\{1,-1\}$.
		\STATE  Randomly draw $B\sim$  Bernoulli$(\eta)$, where $\eta=\frac{|\mathcal{S}|}{\max\{|\mathcal{S}|,\ell\}}$.
		\IF{$B=1$}
		\STATE  Randomly sample one key-value pair $\langle k,v^{*}\rangle$ from $\mathcal{S}$ with discrete uniform distribution. \quad \verb|//sample a non-dummy key-value pair|
		\ELSE
		\STATE Set $v^{*}=0$ and randomly draw $k$ from $\{d+1,\cdots,d^{\prime}\}$ with discrete uniform distribution. \quad \verb|//sample a dummy key-value pair|
		\ENDIF
		\STATE Discretize the value: $v\leftarrow 1$ w.p. $\frac{1+v^{*}}{2}$ or $v\leftarrow -1$ w.p. $\frac{1-v^{*}}{2}$
		\STATE Return $x=\langle k,v\rangle$.
	\end{algorithmic}
	\label{alg:PS}
\end{algorithm}

\textbf{Our Sampling Protocol.} The original Padding-and-Sampling protocol is designed for itemset data and does not work for key-value data. Thus, we modify it to handle the key-value data, shown in Algorithm \ref{alg:PS}, where $d^{\prime}=d+\ell$, $\mathcal{K}^{\prime}=\{1,2,\cdots,d^{\prime}\}$, and parameter $\eta=\frac{|\mathcal{S}|}{\max\{|\mathcal{S}|,\ell\}}$ represents the probability of sampling the non-dummy key-value pairs. The main differences are two-fold. First, we sample one key-value pair instead of one item, and if a dummy key is sampled, we assign a fake value $v^{*}=0$. Second, after sampling, the value is discretized into $1$ or $-1$ for the value perturbation to implement randomized response based mechanism, where the discretization in Line-7 guarantees the unbiasedness because $\mathbb{E}[v]=\frac{1+v^{*}}{2}-\frac{1-v^{*}}{2}=v^{*}$. 

By using  Algorithm \ref{alg:PS}, the large domain size does not affect the probability of sampling a possessed key because it samples from key-value pairs possessed by users. Also, even when the user size is less than the domain size, the frequent keys still have larger probabilities to be sampled by users while only the infrequent keys may not be sampled. Therefore, the two problems of naive sampling protocol in PrivKVM (discussed in Sec \ref{sec:PrivKVM}) can be solved by our advanced one.

For the selection of $\ell$, a smaller $\ell$ will underestimate the frequency thus lead to a large bias, while a larger one will enlarge the variance \cite{wang2018locally}.  Thus, it should balance the tradeoff between bias and variance. A baseline strategy of selecting a good $\ell$ was proposed in \cite{wang2018locally} for itemset data. They set $\ell$ as the 90th percentile of the length of  inputs, where the length distribution is privately estimated from a subset of users. Note that the users are partitioned into multiple groups, where each group participates in only one task (the pre-task to estimate length distribution or the main task to estimate frequency); thus $\epsilon$-LDP in each group guarantees $\epsilon$-LDP for the whole group of users. However, how to select an optimal partition ratio for length distribution estimation (more users in this task can provide more accurate length estimation but leads to fewer users for the main task which impacts frequency and mean estimation) and how to select an optimal percentile (a larger percentile leads to less bias but more variance) are non-trivial tasks. Therefore, in this paper, we select some reasonable $\ell$ for different datasets in experiments for comparing with PrivKVM (which uses naive sampling protocol) and leave the  strategies of finding the optimized partition ratio and percentile for estimating $\ell$ as future work.

\subsection{Perturbation Mechanisms}
\label{sec:perturbation}

\begin{figure}[!t]
    \centering
    \includegraphics[width=3.1in]{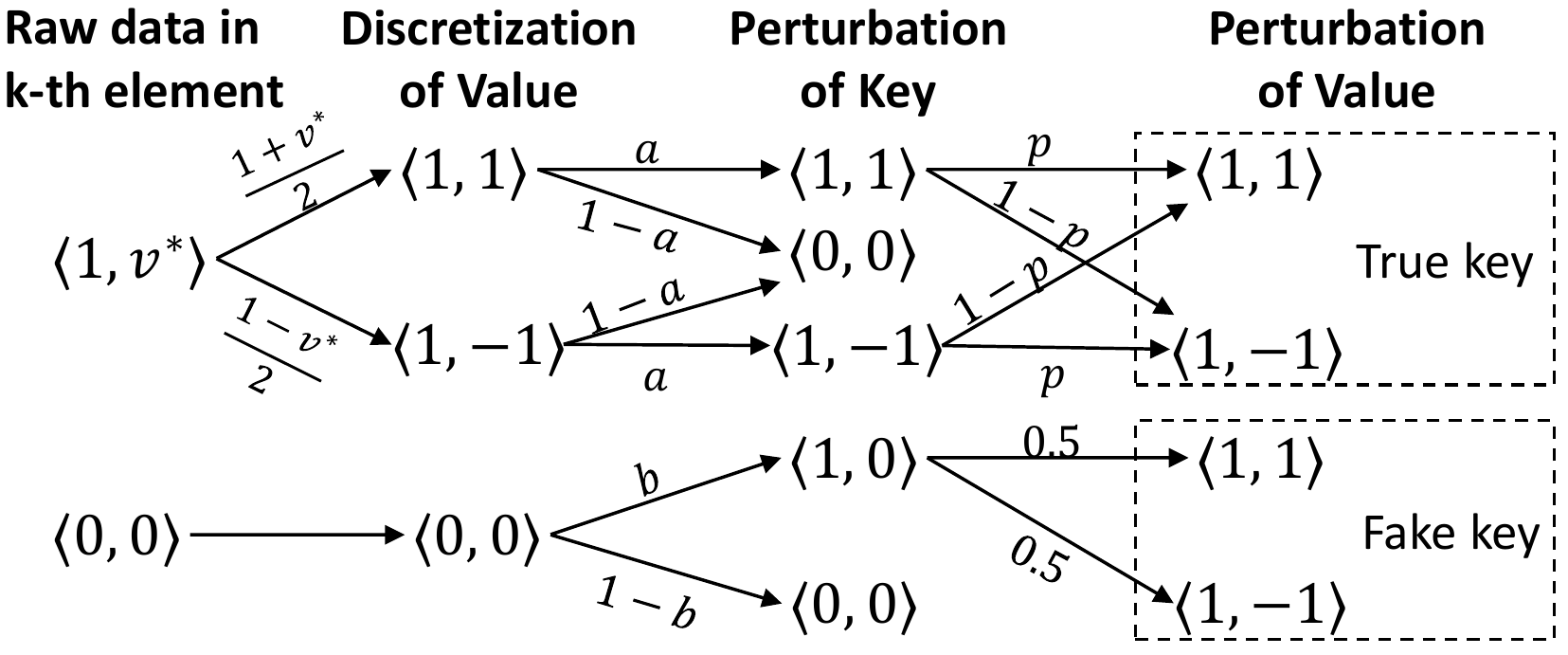}
    \vspace{-3mm}
    \caption{Perturbation of $k$-th element ($\forall k\in\mathcal{K}^{\prime}$) in PCKV-UE.}
    \label{fig:perturbation}
\end{figure}

This subsection corresponds to step \textcircled{\footnotesize 3} in Figure \ref{fig:Diagram}. By  Algorithm \ref{alg:PS}, each user samples one key-value pair $x=\langle k,v\rangle$ as the input of perturbation, where the domain is $k\in\mathcal{K}^{\prime},v\in\{1,-1\}$. If a non-possessed key is reported as possessed (in PCKV-UE), we need to generate fake value. If the original key is perturbed into another one (in PCKV-GRR), the original value is useless for the mean estimation since the original key is not reported to the server. In both cases, we can generate value with discrete uniform distribution to avoid influence of values of different keys. We will show that such a strategy can provide a tighter composition (in Theorem \ref{thm:LDP_UE} and Theorem \ref{thm:LDP_GRR}), which is reflected as a smaller total budget of the composed perturbation than sequential composition. By combining the above idea with sampling protocol for key-value data (Algorithm \ref{alg:PS}) and two basic LDP mechanisms (UE and GRR) in Sec. \ref{sec:LDP mechanism},  we obtain two mechanisms under the PCKV framework: PCKV-UE and PCKV-GRR.

\textbf{PCKV-UE Mechanism.} In Unary Encoding (UE), the original input is encoded as a bit vector, where only the input-corresponding bit is 1 and other bits are 0s, then each bit flips with specified probabilities to generate the output vector.  For key-value data, denote the element in $k$-th position (regarding the key $k$) as $\langle i,v\rangle$ with domain $\{\langle1,1\rangle,\langle1,-1\rangle,\langle0,0\rangle\}$, i.e., the sampled pair $x=\langle k,v \rangle$ is encoded as a vector $\mathbf{x}$, where only the $k$-th element is $\langle 1,\pm 1\rangle$ and others are $\langle 0,0 \rangle$.  Then, the perturbation of key $i\rightarrow i^{\prime}$ in each element can be implemented by $1\rightarrow 1$ w.p. $a$ or $0\rightarrow 1$ w.p. $b$ (where $b\leqslant0.5\leqslant a$). For value perturbation $v\rightarrow v^{\prime}$, we discuss three cases:

Case 1. If  $1\rightarrow 1$, then the value is maintained ($v^{\prime}=v$) with probability $p$ or flipped ($v^{\prime}=-v$) with probability $1-p$. 

Case 2. If $1\rightarrow 0$ or $0\rightarrow 0$, then the output value can be set to $v^{\prime}=0$ because the key $k$ is reported as not possessed.

Case 3. If $0\rightarrow 1$, then the fake value $v^{\prime}=1$ or $v^{\prime}=-1$  are assigned with probability 0.5 respectively.

The discretization and perturbation of PCKV-UE are shown in Figure \ref{fig:perturbation}. For brevity, we use three states $\{1,-1,0\}$ to represent the key-value pairs $\{\langle1,1\rangle,\langle1,-1\rangle,\langle0,0\rangle\}$ in each position of output vector $\mathbf{y}$. If the sampled pair is $x=\langle k,1\rangle$, then only the $k$-th element of the encoded vector $\mathbf{x}$ is 1 (other elements are 0s), and the probability of $\mathbf{y}[k]=1$ is
\begin{align*}
    &\Pr(\mathbf{y}[k]=1|x=\langle k,1\rangle)=\Pr(\mathbf{y}[k]=1|\mathbf{x}[k]=1)=ap
\end{align*}
Similarly, we can compute the perturbation probabilities of other elements in all possible cases, shown in Algorithm \ref{alg:PCKV-UE}. 

\begin{algorithm}[t!]
	\caption{PCKV-UE}
	\footnotesize
	\begin{algorithmic}[1]
		\REQUIRE The set of key-value pairs $\mathcal{S}$, perturbation probabilities $a,b$ and $p$, where $a,p\in[\frac{1}{2},1)$ and $b\in(0,\frac{1}{2}]$.
		\ENSURE Vector $\mathbf{y}\in\{1,-1,0\}^{d^{\prime}}$, where $d^{\prime}=d+\ell$.
		\STATE Sample one key-value pair $x=\langle k,v\rangle$ from $\mathcal{S}$ by Algorithm \ref{alg:PS}.
		\STATE Independently perturb the $k$-th element and other elements ($\forall i\in\mathcal{K}^{\prime}\backslash k$)  
		\begin{align*}
		    \mathbf{y}[k]= 
		    \begin{cases}
		    v, & \text{w.p.}\quad a\cdot p \\
		    -v, & \text{w.p.}\quad a\cdot(1-p) \\
		    0, & \text{w.p.} \quad 1-a
		    \end{cases},\qquad
		    \mathbf{y}[i]= 
		    \begin{cases}
		    1, & \text{w.p.}\quad b/2 \\
		    -1, & \text{w.p.}\quad b/2 \\
		    0, & \text{w.p.}\quad 1-b
		    \end{cases}
		\end{align*}
		\STATE Return vector $\mathbf{y}$.
	\end{algorithmic}
	\label{alg:PCKV-UE}
\end{algorithm}

\textbf{Privacy Analysis of PCKV-UE.}
In PCKV-UE, the key is perturbed by Unary Encoding (UE) with budget $\epsilon_1=\ln\frac{a(1-b)}{b(1-a)}$ (refer to Sec. \ref{sec:LDP mechanism}), and the value is perturbed by Randomized Response (because the discretized value is $1$ or $-1$) with budget $\epsilon_2=\ln\frac{p}{1-p}$ (then $p=\frac{e^{\epsilon_2}}{e^{\epsilon_2}+1}$). Also, the key and value are perturbed in a correlated manner. That is, the value perturbation mechanism depends on both the input key and perturbed key of a user. Intuitively, \emph{correlated perturbation may leak less information  than  independent perturbation}, i.e., the total privacy budget $\epsilon$ can be less than the summation $\epsilon_1+\epsilon_2$. The following theorem shows the tight budget composition of our PCKV-UE mechanism. 

\begin{theorem}[\textbf{Budget Composition of PCKV-UE}]
\label{thm:LDP_UE}
Assume the  privacy budgets for key and value perturbations in PCKV-UE (Algorithm \ref{alg:PCKV-UE}) are $\epsilon_1$ and $\epsilon_2$ respectively, i.e., the perturbation probabilities $a,b,p$ satisfies
\begin{align}
  \frac{a(1-b)}{b(1-a)}=e^{\epsilon_1},\quad
  p=\frac{e^{\epsilon_2}}{e^{\epsilon_2}+1}
\end{align}
then PCKV-UE satisfies LDP with privacy budget
\begin{align}
    \label{equ:epsilon_UE}
    \epsilon=\max\left\{\epsilon_2,~\epsilon_1+\ln[2/(1+e^{-\epsilon_2})]\right\} 
\end{align}
where $\epsilon\leqslant (\epsilon_1+\epsilon_2)$ because of $\epsilon_1\geqslant0$ and $\frac{2}{1+e^{-\epsilon_2}}\leqslant e^{\epsilon_2}$.
\end{theorem}
\begin{proof}
    See Appendix \ref{apx:thm:LDP_UE}.
\end{proof}

\textbf{Interpretation of Theorem \ref{thm:LDP_UE}.} For two different key-value pairs $\langle k_1,v_1\rangle$ and $\langle k_2,v_2\rangle$, where $v_1,v_2\in\{1,-1\}$, the probability ratio of reporting the same output vector $\mathbf{y}$ should be bounded to guarantee LDP. If $k_1=k_2=k$, then the probability ratio only depends on the perturbation of the $k$-th elements $v_1$ and $v_2$ (because other elements are the same, then the corresponding probabilities are canceled out in the ratio), thus the upper bound of the probability ratio is $\frac{ap}{a(1-p)}=e^{\epsilon_2}$. If $k_1\neq k_2$, the ratio depends on both $k_1$-th and $k_2$-th elements, thus the upper bound is $\frac{ap}{b/2}\cdot\frac{1-b}{1-a}=\frac{2e^{\epsilon_1+\epsilon_2}}{e^{\epsilon_2}+1}=e^{\epsilon_1+\ln[2/(1+e^{-\epsilon_2})]}$ (in the case of $\mathbf{y}[k_1]=v_1,\mathbf{y}[k_2]=0$ or $\mathbf{y}[k_1]=0,\mathbf{y}[k_2]=v_2$). Finally, the total privacy budget is the log of the maximum value of the upper bounds in the two cases.

\textbf{Using Theorem \ref{thm:LDP_UE} to Allocate Budget.} Due to the non-linear relationship among $\epsilon$, $\epsilon_1$, and $\epsilon_2$ in Theorem \ref{thm:LDP_UE}, the budget allocation of PCKV-UE is not direct as $\epsilon_2=\epsilon-\epsilon_1$. We discuss the budget allocation in PCKV-UE as follows. Assume $\epsilon>0$ is a given total privacy budget for composed key-value perturbation. According to (\ref{equ:epsilon_UE}), both $\epsilon_1$ and $\epsilon_2$ are less or equal to $\epsilon$. If $\epsilon_1=\epsilon$, we have $\epsilon_2=0$. If $\epsilon_2=\epsilon$, we have $\epsilon_1\leqslant\epsilon-\ln[2/(1+e^{-\epsilon})]=\ln[(e^{\epsilon}+1)/2]$. Therefore, $\epsilon_1$ and $\epsilon_2$ can be allocated by (with respect to a variable $\theta$)
\begin{align}
    \label{equ:epsilon_1&2}
    \epsilon_1=\ln \theta,\quad \epsilon_2=\ln\frac{1}{2\theta e^{-\epsilon}-1},\quad \text{for }\frac{e^{\epsilon}+1}{2}\leqslant \theta <e^{\epsilon}
\end{align}
where $\epsilon_1$ reaches its maximum value when given $\epsilon$ and $\epsilon_2$. The optimized budget allocation, i.e., finding the optimal $\theta$ in \eqref{equ:epsilon_1&2}, will be discussed in Sec. \ref{sec:allocation}.


\textbf{No Tight Budget Composition for PrivKVM.} One may ask if PrivKVM can also be tightly composed like PCKV-UE. Indeed, when the key is perturbed from $0\rightarrow0$ or $1\rightarrow0$ (corresponding to our Case 2) the reported  value must be 0. However, for the case of $0\rightarrow1$ (corresponding to our Case 3), the value is perturbed from the  estimated mean  (discretized as $1$ or $-1$) of the previous iteration with budget $\frac{\epsilon}{2c}$. Therefore, when the output is $\langle1,\cdot\rangle$ (for all rounds), the consumed budget of composed perturbation is $\frac{\epsilon}{2}+c\cdot\frac{\epsilon}{2c}=\epsilon$, which means no tighter composition for PrivKVM.

\begin{algorithm}[t!]
	\caption{PCKV-GRR}
	\footnotesize
	\begin{algorithmic}[1]
		\REQUIRE The set of key-value pairs $\mathcal{S}$, perturbation probabilities $a,p\in[\frac{1}{2},1)$.
		\ENSURE one key-value pair $y^{\prime}=\langle k^{\prime},v^{\prime}\rangle$, where $k^{\prime}\in\mathcal{K}^{\prime}$ and $v^{\prime}\in\{1,-1\}$.
		\STATE  Sample one key-value pair $\langle k,v\rangle$ from $\mathcal{S}$ by Algorithm \ref{alg:PS}.
		\STATE Perturb $\langle k,v\rangle$ into $\langle k^{\prime},v^{\prime}\rangle$ (probability $b=\frac{1-a}{d^{\prime}-1}$)
		\begin{align*}
		    \langle k^{\prime},v^{\prime}\rangle=
		    \begin{cases}
		    \langle k,v \rangle,& \text{w.p.}\quad a\cdot p\\
		    \langle k,-v \rangle,& \text{w.p.}\quad a\cdot (1-p)\\
		    \langle i, 1 \rangle \quad (i\in\mathcal{K}^{\prime}\backslash k), & \text{w.p.}\quad b\cdot 0.5 \\
		    \langle i, -1 \rangle \quad (i\in\mathcal{K}^{\prime}\backslash k), & \text{w.p.}\quad b\cdot 0.5
		    \end{cases}
		\end{align*}
		\STATE Return $y^{\prime}=\langle k^{\prime},v^{\prime}\rangle$.
	\end{algorithmic}
	\label{alg:PCKV-GRR}
\end{algorithm}

\textbf{PCKV-GRR Mechanism.} In GRR, the input is perturbed into another item with specified probabilities, where the input and output have the same domain. In PCKV-GRR, the key is perturbed by GRR with privacy budget $\epsilon_1$, i.e., $k\rightarrow k$ with probability $a=\frac{e^{\epsilon_1}}{e^{\epsilon_1}+d^{\prime}-1}$ and $k\rightarrow i~(i\neq k)$ with probability $b=\frac{1-a}{d^{\prime}-1}$. The value is perturbed by two cases: if $k\rightarrow i~(i\neq k)$, it is perturbed with privacy budget $\epsilon_2$; if $k\neq k^{\prime}$, it is randomly picked from $\{1,-1\}$ with probability 0.5 respectively (similar ideas as in PCKV-UE). The implementation of PCKV-GRR  is shown in Algorithm \ref{alg:PCKV-GRR}.

\textbf{Privacy Analysis of PCKV-GRR.}
Similar to PCKV-UE, the mechanism PCKV-GRR also consumes less privacy budget than $\epsilon_1+\epsilon_2$. Besides the tight composition obtained from the correlated perturbation, \emph{PCKV-GRR would get additional privacy amplification benefit from Padding-and-Sampling},  though our sampling protocol is originally used to avoid privacy budget  splitting (refer to Sec. \ref{sec:sampling}).  


\begin{theorem}[\textbf{Budget Composition of PCKV-GRR}]
\label{thm:LDP_GRR}
Assume the  privacy budgets for key and value perturbation of PCKV-GRR (Algorithm \ref{alg:PCKV-GRR}) are $\epsilon_1$ and $\epsilon_2$ respectively, i.e., the perturbation probabilities $a,b$ and $p$ are
\begin{align}
    \label{equ:abp_GRR}
    a = \frac{e^{\epsilon_1}}{e^{\epsilon_1}+d^{\prime}-1},\quad
    b = \frac{1}{e^{\epsilon_1}+d^{\prime}-1},\quad
    p = \frac{e^{\epsilon_2}}{e^{\epsilon_2}+1}
\end{align}
then PCKV-GRR satisfies LDP with privacy budget
\begin{align}
    \label{equ:epsilon_GRR}
    \epsilon=\ln\left(\frac{e^{\epsilon_1+\epsilon_2}+\lambda}{\min\{e^{\epsilon_1},(e^{\epsilon_2}+1)/2\}+\lambda}\right)
\end{align}
where $\lambda=(\ell-1)(e^{\epsilon_2}+1)/2$.
\end{theorem}
\begin{proof}
    See Appendix \ref{apx:thm:LDP_GRR}.
\end{proof}

\textbf{Interpretation of Theorem \ref{thm:LDP_GRR}.}
According to (\ref{equ:epsilon_GRR}), the total budget $\epsilon$ is a decreasing function of $\lambda$, where $\lambda$ is an increasing function of $\ell$, indicating that a larger $\ell$ provides stronger privacy (smaller $\epsilon$) of PCKV-GRR under the given $\epsilon_1$ and $\epsilon_2$. Also, the above budget composition  has two extreme cases. First, if $\ell=1$, then $\lambda=0$ and (\ref{equ:epsilon_GRR}) reduces to the budget composition of PCKV-UE in  (\ref{equ:epsilon_UE}), which indicates that  the two mechanisms obtain the same benefit (tight budget composition) by adopting the correlated perturbations. Second, if $\epsilon_2=0$, then $\lambda=\ell-1$ and (\ref{equ:epsilon_GRR}) reduces to $\epsilon=\ln\frac{e^{\epsilon_1}+\ell-1}{\ell}$, which is the corresponding result in \cite{wang2018locally} for itemset data. The intuitive reason for such consistency is that the key perturbation will consume all budget when $\epsilon_2=0$; thus, this special case of key-value perturbation can be regarded as item perturbation.

\textbf{No Privacy Benefits from Padding-and-Sampling for PCKV-UE.}
Since Theorem \ref{thm:LDP_UE} is independent of $\ell$ while Theorem \ref{thm:LDP_GRR} is dependent on $\ell$, PCKV-UE does not have the same privacy amplification benefits from Padding-and-Sampling as PCKV-GRR (both of which have been observed in \cite{wang2018locally} for itemset data collection). The main reason is that PCKV-UE outputs a vector that can contain multiple keys (i.e., multiple positions have 1). Take a toy example that only considers the perturbation of key (i.e., $\epsilon_2=0$) with domain $\mathcal{K}=\{1,2,3,4\}$ (then $d=|\mathcal{K}|=4$) and $\ell=2$, where the output domain is $\mathcal{Y}=\{0,1\}^{d+\ell}$. In the worst case that determines the upper bound of the probability ratio, we select two neighboring inputs $\mathcal{S}_1=\{1,2\}$ and $\mathcal{S}_2=\{3,4\}$ (note that LDP considers any set of keys as neighboring for one user) and output vector $\mathbf{y}=[110000]$. No matter which key is sampled from $\mathcal{S}_1$, the probability of reporting $\mathbf{y}$ is the same: $p^{*}=ab(1-b)^4$ (because $\mathbf{x}=[100000]$ or $[010000]$). Considering all sampling cases under sampling rate $\frac{1}{\max\{\ell,|\mathcal{S}_1|\}}$, we have  $\Pr(\mathbf{y}|\mathcal{S}_1)=\frac{1}{\ell}p^{*}\cdot \ell=p^{*}$, which is independent of $\ell$. Similarly, $\Pr(\mathbf{y}|\mathcal{S}_2)=b^2(1-a)(1-b)^3$. Thus, the probability ratio is $\frac{\Pr(\mathbf{y}|\mathcal{S}_1)}{\Pr(\mathbf{y}|\mathcal{S}_2)}=\frac{a(1-b)}{b(1-a)}=e^{\epsilon_1}$, i.e., no privacy benefits from $\ell$. Note that for other $\mathcal{S}_1,\mathcal{S}_2$, and $\mathbf{y}$, the probability ratio might depend on $\ell$, but they are not the worst case that determines the upper bound. For PCKV-GRR, however, the output $y$ can be only one key. In the worst case, we select the above $\mathcal{S}_1$ and $\mathcal{S}_2$ but $y=\{1\}$. Then, $\Pr(y|\mathcal{S}_1)=\frac{1}{\ell}a +(1-\frac{1}{\ell})b$ because if $x=\{1\}$ (resp. $x=\{2\}$) is sampled from $\mathcal{S}_1$, the probability of reporting $y$ is $a$ (resp. $b$), where $a>b$. Also, $\Pr(y|\mathcal{S}_2)=\frac{1}{\ell}b\cdot \ell=b$ (no matter $x=\{3\}$ or $x=\{4\}$ is sampled, the probability of reporting $y$ is $b$). Thus, $\frac{\Pr(y|\mathcal{S}_1)}{\Pr(y|\mathcal{S}_2)}=1+\frac{a/b-1}{\ell}\leqslant\frac{a}{b}=e^{\epsilon_1}$, where a larger $\ell$ will reduce this ratio (i.e., privacy amplification). 

Theorem \ref{thm:LDP_UE} and Theorem \ref{thm:LDP_GRR} provide a tighter bound on the total privacy guarantee than the sequential composition ($\epsilon=\epsilon_1+\epsilon_2$). However, in practice, the budgets are determined in a reverse way: given $\epsilon$ (a constant), we need to allocate the corresponding $\epsilon_1$ and $\epsilon_2$ before any perturbation. In Sec. \ref{sec:allocation}, we will discuss the  optimized privacy budget allocation (i.e., how to determine $\epsilon_1$ and $\epsilon_2$ when $\epsilon$ is given) by minimizing the estimation error that is analyzed in Sec. \ref{sec:calibratoin}. In summary, both the tight budget composition and optimized budget allocation in our scheme will improve the privacy-utility tradeoff. Note that PrivKVM \cite{ye2019privkv} simply allocates the privacy budget with $\epsilon_1=\epsilon_2=\epsilon/2$ by sequential composition (Theorem \ref{thm:seq_compo}).

\subsection{Aggregation and Estimation}
\label{sec:calibratoin}

This subsection corresponds to step \textcircled{\footnotesize 4} \textcircled{\footnotesize 5} in Figure \ref{fig:Diagram}. Intuitively, the value mean of a certain key can be estimated by the ratio between the summation of all true values and the count of values regarding this key; however, the fake values affect both the summation and the count. In PrivKVM \cite{ye2019privkv}, since the count of values includes the fake ones, the mean of fake values should be close to the true mean to guarantee the unbiasedness of estimation. Therefore, a large number of iterations are needed to make the fake values approach the true mean. In our scheme, however, the fake values have expected zero summation because they are assigned as $-1$ or $1$ with probability 0.5 respectively. Therefore, we can use the estimated frequency to approach the count of truly existing values, thus only one round is needed.

\textbf{Aggregation.}
After all users upload their outputs to the server, the server will count the number of $1$'s and $-1$'s that supports  $k\in\mathcal{K}$ in output, denoted as $n_1$ and $n_2$ respectively (the subscript $k$ is omitted for brevity). Since the outputs of the proposed two mechanisms have different formats, the server computes $n_1 = Count(\mathbf{y}[k]=1)$ and $n_2 = Count(\mathbf{y}[k]=-1)$ in PCKV-UE, or computes $n_1 = Count(y^{\prime}=\langle k,1\rangle)$ and $n_2 = Count(y^{\prime}=\langle k,-1\rangle)$ in PCKV-GRR. Then, $n_1$ and $n_2$ will be calibrated to estimate the frequency and mean of key $k\in\mathcal{K}$.

\textbf{Baseline Estimation Method.}
For frequency estimation, we use the estimator in \cite{wang2018locally} for itemset data, which is shown to be unbiased when each user's itemset size is no more than $\ell$. Since $n_1+n_2$ is the observed count of users that possess the key, we have the following equivalent \emph{frequency estimator}
\begin{align}
    \label{equ:hat_f}
    \hat{f}_k=\frac{(n_1+n_2)/n-b}{a-b}\cdot\ell
\end{align}
For mean estimation, since our mechanisms generate the fake values as $-1$ or $1$ with probability 0.5 respectively (i.e., the expectation is zero), they have no contribution to the value summation statistically. Therefore, we can estimate the value mean by dividing the summation  with the count of real keys. According to Randomized Response (RR) in Sec. \ref{sec:LDP mechanism}, the calibrated summation is $\frac{n_1-n(1-p)}{2p-1}-\frac{n_2-n(1-p)}{2p-1}=\frac{n_1-n_2}{2p-1}$.  The count of real keys which are still reported as possessed can be approximated by $n\hat{f}_ka/\ell$ because the sampling rate is $1/\ell$ and real keys are reported as possessed with probability $a$. Therefore, the corresponding  \emph{mean estimator} is  
\begin{align}
    \label{equ:hat_m}
    \hat{m}_k=\frac{(n_1-n_2)/(2p-1)}{n\hat{f}_ka/\ell}=\frac{(n_1-n_2)(a-b)}{ a(2p-1)(n_1+n_2-nb)}
\end{align}
The following theorem analyzes the expectation and variance of our estimators in (\ref{equ:hat_f}) and (\ref{equ:hat_m}) when each user has no more than $\ell$ key-value pairs (the same condition as in \cite{wang2018locally}).

\begin{theorem}[\textbf{Estimation Error Analysis}]
\label{thm:estimation}
If the padding length $\ell\geqslant |\mathcal{S}_u|$ for all user $u\in\mathcal{U}$; then, for frequency and mean estimators in (\ref{equ:hat_f}) and (\ref{equ:hat_m}) of $k\in\mathcal{K}$, $\hat{f}_k$ is unbiased, i.e., $\mathbb{E}[\hat{f}_k]=f_k^{*}$, and their expectation and variance are 
\begin{align}
    \label{equ:Var[f_k]}
    &\text{Var}[\hat{f}_k]= \frac{\ell^2 b(1-b)}{n(a-b)^2} + \frac{\ell\cdot f_k^{*}(1-a-b)}{n(a-b)}\\
    \label{equ:E[m_k]}
    &\mathbb{E}[\hat{m}_k]\approx m_k^{*}\left[1+\frac{ (1-b-\delta)b }{n\delta^2}\right]\\
    \label{equ:Var[m_k]}
    &\text{Var}[\hat{m}_k]\lesssim\frac{b+\delta}{n\gamma^2} + \frac{b(1-b)-\delta}{n\delta^2} \cdot{m_k^{*}}^2
\end{align}
where parameters $\delta$ and $\gamma$ are defined by
\begin{align}
    \label{equ:delta}
    \delta=(a-b)f_k^{*}/\ell,\quad \gamma=a(2p-1)f_k^{*}/\ell
\end{align}
The variance in (\ref{equ:Var[m_k]}) is an approximate upper bound and the approximation in (\ref{equ:E[m_k]}) and (\ref{equ:Var[m_k]}) is from Taylor expansions.

\end{theorem}
\begin{proof}
    See Appendix \ref{apx:thm:estimation_UE}. Note that Theorem \ref{thm:estimation} works for both PCKV-UE and PCKV-GRR.
\end{proof}

\textbf{Pros and Cons of the Baseline Estimator.}
The baseline estimation method estimates frequency and mean by (\ref{equ:hat_f}) and (\ref{equ:hat_m}) respectively. According to (\ref{equ:E[m_k]}) and (\ref{equ:Var[m_k]}), for non-zero constants $\delta$ and $\gamma$, when the user size $n\rightarrow+\infty$, we have $\mathbb{E}[\hat{m}_k]-m_k^{*}=\frac{(1-b-\delta)bm_k^{*}}{n\delta^2}\rightarrow 0$ (i.e., the bias of $\hat{m}_k$ is 
progressively approaching 0) and $\text{Var}[\hat{m}_k]\rightarrow 0$, which means $\hat{m}_k$ converges in probability to the true mean $m_k^{*}$. However, when $\frac{1}{n(f^{*}_k/\ell)^2}$ is not small, the large bias and large variance would make the estimated mean $\hat{m}_k$ far away from the true mean, even out of the bound $[-1,1]$. Similarly, if $\text{Var}[\hat{f}_k]$ in (\ref{equ:Var[f_k]}) is not very small, then for $f_k^{*}\rightarrow 0$ or $f_k^{*}\rightarrow 1$,  the estimated frequency $\hat{f}_k$ may also be outside the bound $[0,1]$. Hence, these outliers need further correction to reduce the estimation error.

\begin{algorithm}[t!]
	\caption{Aggregation and Estimation with Correction}
	\footnotesize
	\begin{algorithmic}[1]
		\REQUIRE Outputs of all users, domain of keys $\mathcal{K}$, perturbation probabilities $a,b,p$ and padding length $\ell$.
		\ENSURE Frequency and mean estimation $\hat{f}_k$ and $\hat{m}_k$ for all $k\in\mathcal{K}$.
		\FOR{$k\in\mathcal{K}$}
		\STATE Count the number of supporting $1$'s and $-1$'s for key $k$ in outputs from all users, denoted as $n_1$ and $n_2$.
		\STATE  Compute $\hat{f}_k$ by (\ref{equ:hat_f}) and correct it into $[1/n,1]$.
		\STATE Compute $\hat{n}_1$ and $\hat{n}_2$ by (\ref{equ:hat_n}), and correct them into  $[0,n\hat{f}_k/\ell]$.
		\STATE Compute $\hat{m}_k$ by (\ref{equ:hat_m_new}).
		\ENDFOR
		\STATE Return $\hat{f}_k$ and $\hat{m}_k$, where $k\in\mathcal{K}$.
	\end{algorithmic}
	\label{alg:Estimation}
\end{algorithm}

\textbf{Improved Estimation with Correction.} 
Since the value perturbation depends on the output of key perturbation,  we first correct the result of frequency estimation. Considering the corrected frequency cannot be 0 (otherwise the mean estimation will be infinity), we clip the frequency values using the range  $[1/n,1]$, i.e., set the outliers less than $1/n$ to $1/n$ and outliers larger than $1$ to $1$. For the mean estimation, denote the true counts of sampled key-value pair $x=\langle k,1\rangle$ and $x=\langle k,-1\rangle$ (the output of Algorithm \ref{alg:PS}) of all users as $n_1^{*}$ and $n_2^{*}$ respectively (the subscript $k$ is omitted for brevity). Then we have the following lemma for the estimation of $n_1^{*}$ and $n_2^{*}$.
\begin{lemma}
\label{lem:hat_n}
The unbiased estimators of $n_1^{*}$ and $n_2^{*}$ are
\begin{align}
    \label{equ:hat_n}
    \begin{bmatrix}
    \hat{n}_1\\
    \hat{n}_2
    \end{bmatrix}=A^{-1}
    \begin{bmatrix}
    n_1-nb/2\\
    n_2-nb/2
    \end{bmatrix}, \text{ where }
    A=\left[\begin{smallmatrix}
    ap-\frac{b}{2} & a(1-p)-\frac{b}{2}\\
    a(1-p)-\frac{b}{2} & ap-\frac{b}{2}
    \end{smallmatrix}\right]
\end{align}
\end{lemma}
\begin{proof}
    See Appendix \ref{apx:lem:hat_n}.
\end{proof}

Note that Lemma \ref{lem:hat_n} works for both PCKV-UE and PCKV-GRR. According to (\ref{equ:hat_n}), we have
\begin{align*}
    \hat{n}_1-\hat{n}_2=
    \begin{bmatrix}
    1 & -1
    \end{bmatrix}
    A^{-1}
    \begin{bmatrix}
    n_1-nb/2\\
    n_2-nb/2
    \end{bmatrix}
    =\frac{n_1-n_2}{a(2p-1)}
\end{align*}
then $\hat{m}_k$ in (\ref{equ:hat_m}) can be represented by $\hat{n}_1-\hat{n}_2$ and $\hat{f}_k$ in (\ref{equ:hat_f})
\begin{align}
    \label{equ:hat_m_new}
    \hat{m}_k =\ell(\hat{n}_1-\hat{n}_2)/(n\hat{f}_k) 
\end{align}
which means $n_1^{*}+n_2^{*}$ (the supporting number of $1$ and $-1$ for key $k\in\mathcal{K}$) is estimated by $n\hat{f}_k/\ell$. Therefore, $\hat{n}_1$ and $\hat{n}_2$ should be bounded by $[0,n\hat{f}_k/\ell]$. The aggregation and estimation mechanism (with correction) is shown in Algorithm \ref{alg:Estimation}, where the  difference between PCKV-UE and PCKV-GRR is only on the aggregation step, which is caused by the different types of output (one is a vector, another is a key-value pair).

\subsection{Optimized Privacy Budget Allocation}
\label{sec:allocation}
In this section, we discuss how to optimally allocate budgets $\epsilon_1$ and $\epsilon_2$ given the total privacy budget $\epsilon$, which corresponds to step \textcircled{\footnotesize 1} in Figure \ref{fig:Diagram}. 
The budget composition (Theorem \ref{thm:LDP_UE} and Theorem \ref{thm:LDP_GRR}) provides the relationship among $\epsilon$, $\epsilon_1$, and $\epsilon_2$. Intuitively, when the total privacy budget $\epsilon$ is given, we can find the optimal $\epsilon_1$ and $\epsilon_2$ that satisfy the budget composition by solving an optimization problem of minimizing the combined Mean Square Error (MSE) of frequency and mean estimations, i.e., $\alpha\cdot \text{MSE}_{\hat{f}_k}+\beta\cdot \text{MSE}_{\hat{m}_k}$.  However,  from Theorem \ref{thm:estimation}, $\text{Var}[\hat{f}_k]$ and $\text{Var}[\hat{m}_k]$ depend on $f_k^{*}$ and $m_k^{*}$, whose true values or even the approximate values are unknown in the budget allocation stage (before any perturbation).  Therefore, in the following, we simplify this optimization problem to obtain a practical budget allocation solution with closed-form. Note that a larger $\epsilon_1$ can benefit both frequency and mean estimations, but it restricts $\epsilon_2$ (which affects mean estimation) due to limited $\epsilon$.

\textbf{Problem Simplification of Budget Allocation.}
In this paper, we use Mean Square Error (MSE) to evaluate utility mechanisms, i.e., the less MSE the better utility. Note that the MSE of an estimator $\hat{\theta}$ can be calculated by the summation of variance and the square of its bias
\begin{align}
    \label{equ:MSE}
    \text{MSE}_{\hat{\theta}}=\text{Var}[\hat{\theta}]+\text{Bias}^2=\text{Var}[\hat{\theta}]+(\mathbb{E}[\hat{\theta}]-\theta)^2
\end{align}
When MSE is relatively large, the estimators will be corrected by the improved estimation in Algorithm \ref{alg:Estimation}. Therefore, we mainly consider minimizing  MSE when it is relatively small, i.e., $(2p-1)$ and $(a-b)$ are not very small, and $n$ (the number of users) is very large. Since $f_k^{*}\ll 1$ for most cases in real-world data, we have
$\delta=(a-b)f_k^{*}/\ell\ll 1$. Denote 
\begin{align}
    \label{equ:gh}
    \mu=\frac{\ell^2}{n{f_k^{*}}^2},\quad
    g=\frac{b}{a^2(2p-1)^2},\quad
    h = \frac{(1-b)b}{(a-b)^2}
\end{align}
The MSEs in Theorem \ref{thm:estimation} can be approximated by
\begin{align}
    \text{MSE}_{\hat{f}_k}&=\text{Var}[\hat{f}_k]\approx\ell^2\cdot h/n\\
    \label{equ:MSE_m}
     \text{MSE}_{\hat{m}_k}&\approx\mu[g+(\mu h+1)h{m_k^{*}}^2]
     \approx\mu(g+h\cdot{m_k^{*}}^2)
\end{align}
where $\mu\ll 1$ with a large $n$. Note that $\text{MSE}_{\hat{m}_k}$ dominates $\text{MSE}_{\hat{f}_k}$ because $\frac{\ell^2}{n}/\mu={f_k^{*}}^2\ll 1$. It is caused by the distinct sample size of the two estimations, i.e., frequency is estimated from all users (with user size $n$), while the value mean is estimated from the users who possess a certain key (with user size $nf_k^{*}$). Therefore, our objective function $\alpha\cdot \text{MSE}_{\hat{f}_k}+\beta\cdot \text{MSE}_{\hat{m}_k}$ mainly depends on  $\text{MSE}_{\hat{m}_k}$ when $\alpha$ and $\beta$ are in the same magnitude. Motivated by this observation, we focus on minimizing $\text{MSE}_{\hat{m}_k}$ to obtain the optimized budget allocation. Note that $\text{MSE}_{\hat{f}_k}$ only depends on $\epsilon_1$ (the more $\epsilon_1$ the less $\text{MSE}_{\hat{f}_k}$), while $\text{MSE}_{\hat{m}_k}$  depends on both $\epsilon_1$ and $\epsilon_2$. However, if $\epsilon_1$ approaches to the maximum, which corresponds to the minimum $\text{MSE}_{\hat{f}_k}$, then $\epsilon_2=0$ and $\text{MSE}_{\hat{m}_k}\rightarrow \infty$. In the following, we discuss the optimized privacy budget allocation with minimum $\text{MSE}_{\hat{m}_k}$ in PCKV-UE and PCKV-GRR. 

\textbf{Budget Allocation of PCKV-UE.} In UE-based mechanisms, the Optimized Unary Encoding (OUE) \cite{wang2017locally} was shown to have the minimum MSE of frequency estimation under the same privacy budget. Accordingly, the OUE-based perturbation probabilities for key-value perturbation are
\begin{align}
    \label{equ:abp_UE}
    a=1/2,\quad
    b=1/(e^{\epsilon_1}+1),\quad
    p=e^{\epsilon_2}/(e^{\epsilon_2}+1)
\end{align}
where the values of $a$ and $b$ correspond to the minimum $\text{MSE}_{\hat{f}_k}$ under a given $\epsilon_1$ (budget for key perturbation). Furthermore, by minimizing $\text{MSE}_{\hat{m}_k}$, we have the following optimized budget allocation of PCKV-UE.

\begin{lemma}[\textbf{Optimized Budget Allocation of PCKV-UE}]
\label{lem:budget_UE}
For a total privacy budget $\epsilon$, the optimized budget allocation for key and value perturbations can be approximated by
\begin{align}
    \label{equ:budget_UE}
    \epsilon_1=\ln[(e^{\epsilon}+1)/2],\quad
    \epsilon_2=\epsilon
\end{align}
\end{lemma}
\begin{proof}
    See Appendix \ref{apx:lem:budget_UE}.
\end{proof}


\textbf{Interpretation of Lemma \ref{lem:budget_UE}.}
According to the budget allocation of PCKV-UE in \eqref{equ:epsilon_1&2}, $\epsilon_1$ is an increasing function of $\theta$, while  $\epsilon_2$ and the summation $\epsilon_1+\epsilon_2=\ln\frac{\theta}{2\theta e^{-\epsilon}-1}$ are decreasing functions of $\theta$. From \eqref{equ:budget_UE}, $\epsilon_1$ and $\epsilon_2$ are optimally allocated at $\theta=\frac{e^{\epsilon}+1}{2}$ (the minimum value), which corresponds to the maximum summation $\epsilon_1+\epsilon_2$. Moreover, under the optimized budget allocation, the two values in the max operation in \eqref{equ:epsilon_UE} equal to each other, i.e., $\epsilon_2=\epsilon_1+\ln[2/(1+e^{-\epsilon_2})]=\epsilon$, which indicates that the budgets are fully allocated.


\begin{figure}[!t]
    \centering
    \includegraphics[width=3.3in]{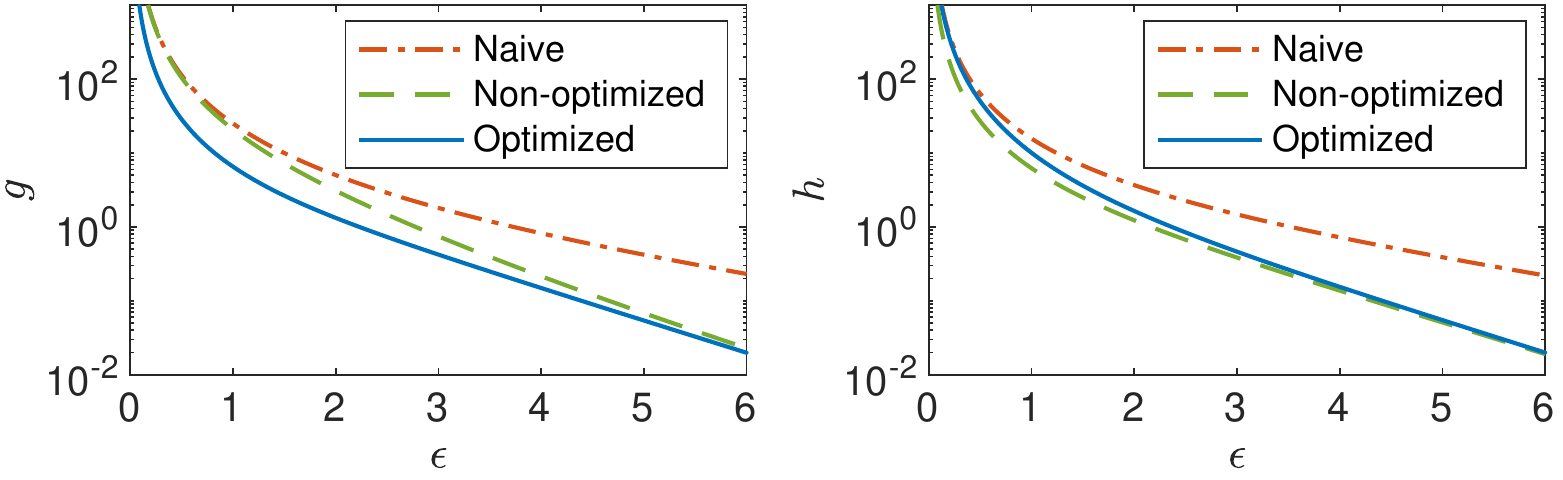}
    \vspace{-4mm}
    \caption{Comparison of $g$ and $h$  under three budget allocation methods for PCKV-UE, where $\text{MSE}_{\hat{m}_k}\approx\mu(g+h\cdot{m_k^{*}}^2)$.}
    \label{fig:allocation}
\end{figure}

\textbf{Comparison with Other Allocation Methods.} 
In order to show the advantage of our \emph{optimized allocation} in (\ref{equ:budget_UE}), we compare it with two alternative methods. The first one is \emph{naive allocation} with $\epsilon_1=\epsilon_2=\epsilon/2$ by sequential composition (which is used in PrivKVM). The second one is \emph{non-optimized allocation} with 
\begin{align}
    \label{equ:non-optimized}
    \epsilon_1=\ln[(e^\epsilon+e^{\epsilon/2})/2],\quad
    \epsilon_2=\epsilon/2
\end{align}
which sets $\epsilon_2$ as $\epsilon/2$ and computes $\epsilon_1$ by our tight budget composition (Theorem \ref{thm:LDP_UE}). Considering $\text{MSE}_{\hat{m}_k}\approx\mu(g+h\cdot{m_k^{*}}^2)$ in \eqref{equ:MSE_m}, we compare  parameters $g$ and $h$ (with respect to $\epsilon$) under above three budget allocation methods, shown in Figure \ref{fig:allocation}. We can observe that the optimized allocation has a much smaller $g$ than the other two, though a little bit larger $h$ than the non-optimized one, which is caused by the property that $h$ is a monotonically decreasing function of $\epsilon_1$, while $\epsilon_1$ and $\epsilon_2$ restrict each other. Note that in our optimized allocation, the decrement of $g$ dominates the increment of $h$. Thus, $\text{MSE}_{\hat{m}_k}$ in (\ref{equ:MSE_m}) will be greatly reduced since ${m_k^{*}}^2\leqslant 1$.

\textbf{Budget Allocation of PCKV-GRR.} 
According to the budget composition (Theorem \ref{thm:LDP_GRR}) of PCKV-GRR, a larger padding length $\ell$ will further improve the privacy-utility tradeoff of key-value  perturbation. Thus, given fixed total budget, the allocated budget for key (or value) perturbation can be larger (i.e., less noise will be added) under a larger $\ell$. The following lemma shows the optimized budget allocation (related to $\ell$) of PCKV-GRR  with minimum $\text{MSE}_{\hat{m}_k}$.
\begin{lemma}[\textbf{Optimized Budget Allocation of PCKV-GRR}]
\label{lem:budget_GRR}
For a total privacy budget $\epsilon$, the optimized budget allocation for key and value perturbation can be approximated by 
\begin{align}
    \label{equ:budget_GRR}
    \epsilon_1=\ln\left[\ell\cdot(e^\epsilon-1)/2+1\right],\quad
    \epsilon_2=\ln\left[\ell\cdot(e^\epsilon-1)+1\right]
\end{align}
\end{lemma}
\begin{proof}
 See Appendix \ref{apx:lem:budget_GRR}.
\end{proof}

According to (\ref{equ:abp_GRR}) and  (\ref{equ:budget_GRR}), with a given total budget $\epsilon$, the perturbation probabilities in PCKV-GRR are
\begin{align}
    \label{equ:abp_GRR_opt}
    a=\frac{\ell(e^{\epsilon}-1)+2}{\ell(e^{\epsilon}-1)+2d^{\prime}},~
    b=\frac{1-a}{d^{\prime}-1},~
    p=\frac{\ell(e^{\epsilon}-1)+1}{\ell(e^{\epsilon}-1)+2}
\end{align}
where $d^{\prime}=d+\ell$. Note that when $\ell=1$, the optimized budget allocation in \eqref{equ:budget_GRR} reduces to the case of PCKV-UE in \eqref{equ:budget_UE}.


\begin{figure}[!t]
    \centering
    \includegraphics[width=3.4in]{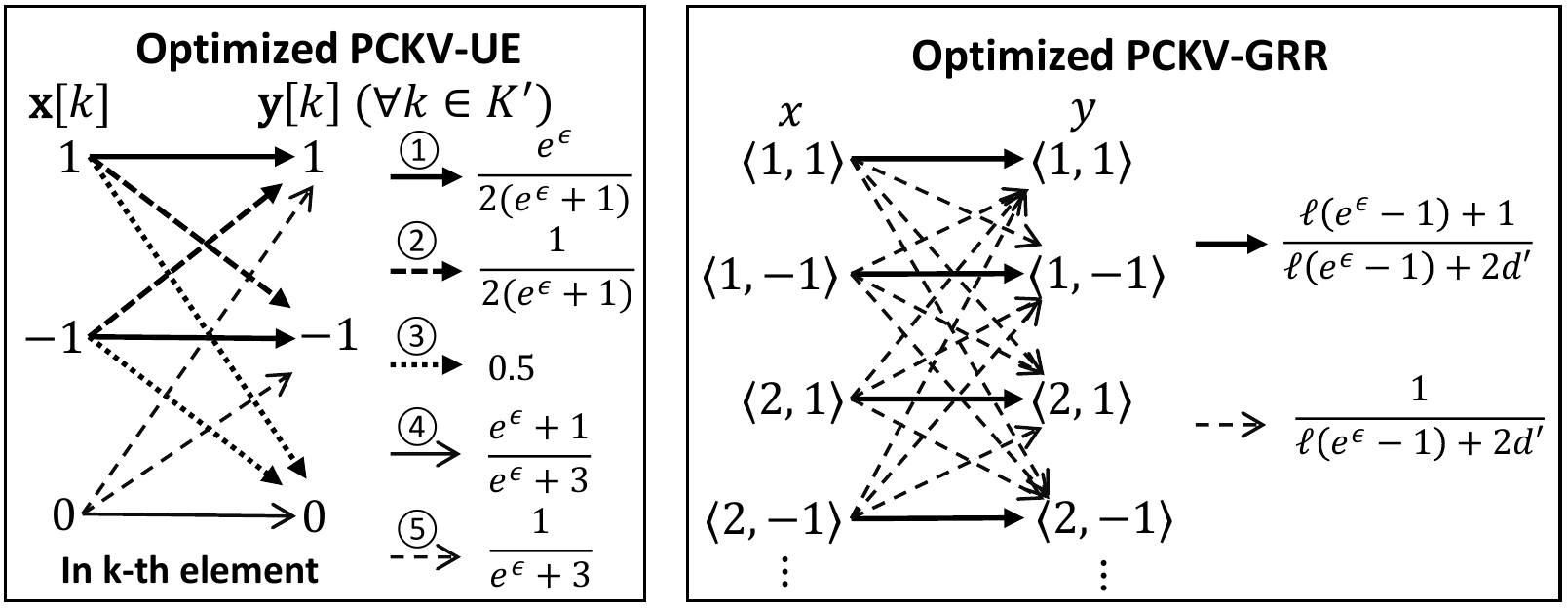}
    \vspace{-8mm}
    \caption{Diagram of our optimized protocols (different types of arrows represent perturbations with different probabilities).}
    \label{fig:optimized}
\end{figure}

\textbf{Interpretation of the Optimized Protocols.} 
Under the optimized budget allocation (Lemma \ref{lem:budget_UE} and Lemma \ref{lem:budget_GRR}), the perturbation probabilities of proposed protocols are shown in Figure \ref{fig:optimized}. In optimized PCKV-UE, for two different input vectors $\mathbf{x}_1$ and $\mathbf{x}_2$ (encoded from the sampled key-value pairs), no matter they differ in one element (i.e., the sampled ones have the same key but different values) or differ in two elements (i.e., the sampled ones have different keys), the upper bound of the probability ratio of outputting the same vector $\mathbf{y}$ is the same, i.e., {\scriptsize $\frac{\textcircled{\scriptsize1}}{\textcircled{\scriptsize 2}}=\frac{\textcircled{\scriptsize 1}}{\textcircled{\scriptsize 5}}\cdot\frac{\textcircled{\scriptsize 4}}{\textcircled{\scriptsize 3}}=$} $e^{\epsilon}$ in Figure \ref{fig:optimized}. In optimized PCKV-GRR, two of three different perturbation probabilities in Algorithm \ref{alg:PCKV-GRR} equal with each other, i.e., $a(1-p)=b\cdot 0.5$ in the optimized solution. Also, the optimized PCKV-GRR can be regarded as the equivalent version of general GRR with doubled domain size (each key can have two different values), which can provide good utility on estimating the counts of $\langle k,1\rangle$ and $\langle k,-1\rangle$, say $n_{k1}$ and $n_{k2}$, where the mean of key $k$ can be estimated by $\frac{n_{k1} - n_{k2}}{n_{k1} + n_{k2}}$.

From the previous analysis,  PCKV-GRR can get additional benefit from sampling, thus it will outperform PCKV-UE for a large $\ell$.  On the other hand, the performance of PCKV-UE is independent of the domain size $d$, thus it will have less MSE than PCKV-GRR when $d$ is very large. Therefore, the two mechanisms are suitable for different cases.
By comparing parameters $g$ and $h$ in \eqref{equ:gh} of PCKV-UE and PCKV-GRR respectively,  for a smaller $\text{MSE}_{\hat{f}_k}$ (i.e., a smaller $h$), if $2(d-1)>\ell(4\ell-1)(e^\epsilon+1)$, then PCKV-UE is better; otherwise, PCKV-GRR is better. For a smaller $\text{MSE}_{\hat{m}_k}$ (i.e., a smaller $g$ approximately), if $2d>\ell\left(\frac{4\ell(e^\epsilon+1)}{e^\epsilon+3}-1\right)(e^\epsilon+1)$, then PCKV-UE is better; otherwise, PCKV-GRR is better. These can be observed in simulation results (Sec. \ref{sec:evaluation}).

\section{Evaluation}
\label{sec:evaluation}

In this section, we evaluate the performance of our proposed mechanisms (PCKV-UE and PCKV-GRR) and compare them with the existing mechanisms (PrivKVM \cite{ye2019privkv} and KVUE \cite{sun2019conditional}). We note that although KVUE \cite{sun2019conditional} is not formally published, we still implemented it with our best effort and included it for comparison purposes.

\textbf{Mechanisms for Comparison.}
In PrivKVM \cite{ye2019privkv}, the number of iterations is set as $c=1$ because we observe that PrivKVM with a large number of iterations $c$ will have bad utility, which is caused by the small budget $\frac{\epsilon}{2c}$ and thus large variance of value perturbation in the last iteration (even though the result is theoretically unbiased). However, implementing PrivKVM with virtual iterations to predict the mean estimation of remaining iterations can avoid budget split \cite{ye2019privkv}. Thus, we also evaluate PrivKVM with one \emph{real} iteration and five \emph{virtual} iterations (\verb|1r5v|).  In \cite{sun2019conditional}, multiple mechanisms are proposed to improve the performance of PrivKVM, where the most promising one is KVUE (which uses the same sampling protocol as in PrivKVM). Note that the original KVUE does not have corrections for mean estimation. For a fair comparison with PrivKVM, PCKV-UE, and PCKV-GRR (outliers are corrected in these mechanisms), we use the similar correction strategy used in PrivKVM for KVUE.


\textbf{Datasets.} In this paper, we evaluate two existing mechanisms (PrivKVM \cite{ye2019privkv} and KVUE \cite{sun2019conditional}) and our mechanisms (PCKV-UE and PCKV-GRR) by synthetic datasets and real-world datasets. In synthetic datasets, the number of users is $n=10^6$, and the domain size is $d=100$, where each user only has one key-value pair (i.e., $\ell=1$), and both the possessed key of each user and the value mean of keys satisfy Uniform (or Gaussian) distribution. The Gaussian distribution is generated with $\mu=0,\sigma_\text{key}=50,\sigma_\text{mean}=1$, where samples  outside  the domain ($\mathcal{K}$ or $\mathcal{V}=[-1,1]$) are discarded. In real-world datasets, each user may have multiple key-value pairs, i.e., $\ell>1$ (how the selection of $\ell$ affects the estimation accuracy has been discussed in Sec. \ref{sec:sampling}).  Table \ref{tab:data} summarizes the parameters of four real-world rating datasets (obtained from public data sources) with different domain sizes and data distributions. The item-rating corresponds to key-value, and all ratings are linearly normalized into $[-1,1]$.

\begin{table}[t!]
    \footnotesize
    \centering
    \caption{Real-World Datasets}
    \vspace{-4mm}
    \begin{tabular}{|c|ccc|c|}
    \hline
    \textbf{Datasets} & \textbf{\# Ratings} & \textbf{\# Users} &  \textbf{\# Keys} & \textbf{Selected $\ell$}\\
    \hline
    E-commerce \cite{ecommerce} & 23,486 & 23,486 & 1,206 & 1\\
    \hline
    Clothing \cite{clothing} & 192,544 & 105,508 & 5,850 & 2\\ 
    \hline
    Amazon \cite{amazon} & 2,023,070 & 1,210,271 & 249,274 & 2\\
    \hline
    Movie \cite{movie}  & 20,000,263 & 138,493& 26,744 & 100\\
    \hline
    \end{tabular}
    \label{tab:data}
\end{table}

\textbf{Evaluation Metric.} We evaluate both the frequency and mean estimation by the \emph{averaged} Mean Square Error (MSE) among all keys or a portion of keys 
\begin{align*}
    \text{MSE}_\text{freq} = \frac{1}{|\mathcal{X}|}\sum_{i\in \mathcal{X}}(\hat{f}_i-f_i^{*})^2,\quad
    \text{MSE}_\text{mean} = \frac{1}{|\mathcal{X}|}\sum_{i\in \mathcal{X}}(\hat{m}_i-m_i^{*})^2
\end{align*}
where  $f_i^{*}$ and $m_i^{*}$ (resp. $\hat{f}_i$ and $\hat{m}_i$) are the true (resp. estimated) frequency and mean, and $\mathcal{X}$ is a subset of the domain $\mathcal{K}$ (the default $\mathcal{X}$ is $\mathcal{K}$). We also consider $\mathcal{X}$ as the set of top $N$ frequent keys (such as top 20 or top 50) because we usually only care about the estimation results of frequent keys. Also, infrequent keys do not have enough samples to obtain the accurate estimation of value mean. All MSE results are averaged with five repeats.

\subsection{Synthetic Data}

\begin{figure}[!t]
    \footnotesize
    \centering
    \includegraphics[width=3.2in]{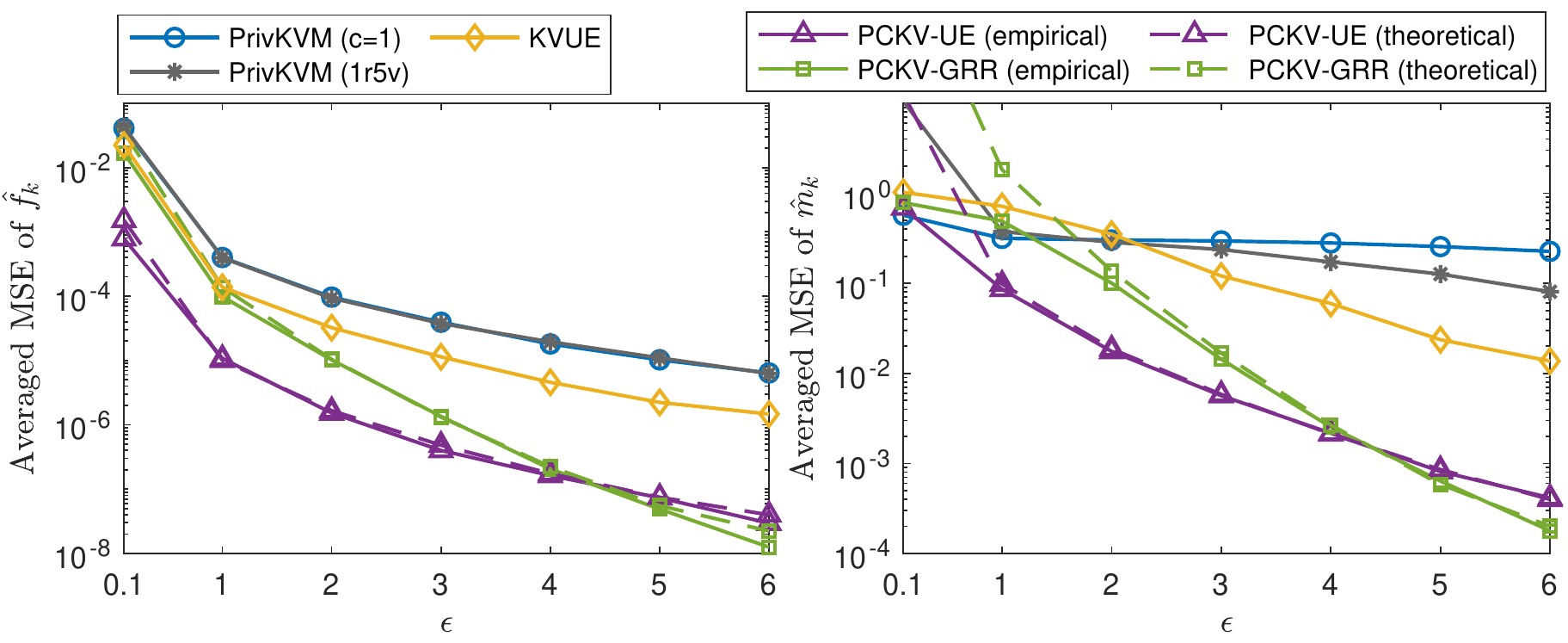}
    (a) Uniform distribution (MSE is averaged of all keys)
    \includegraphics[width=3.2in]{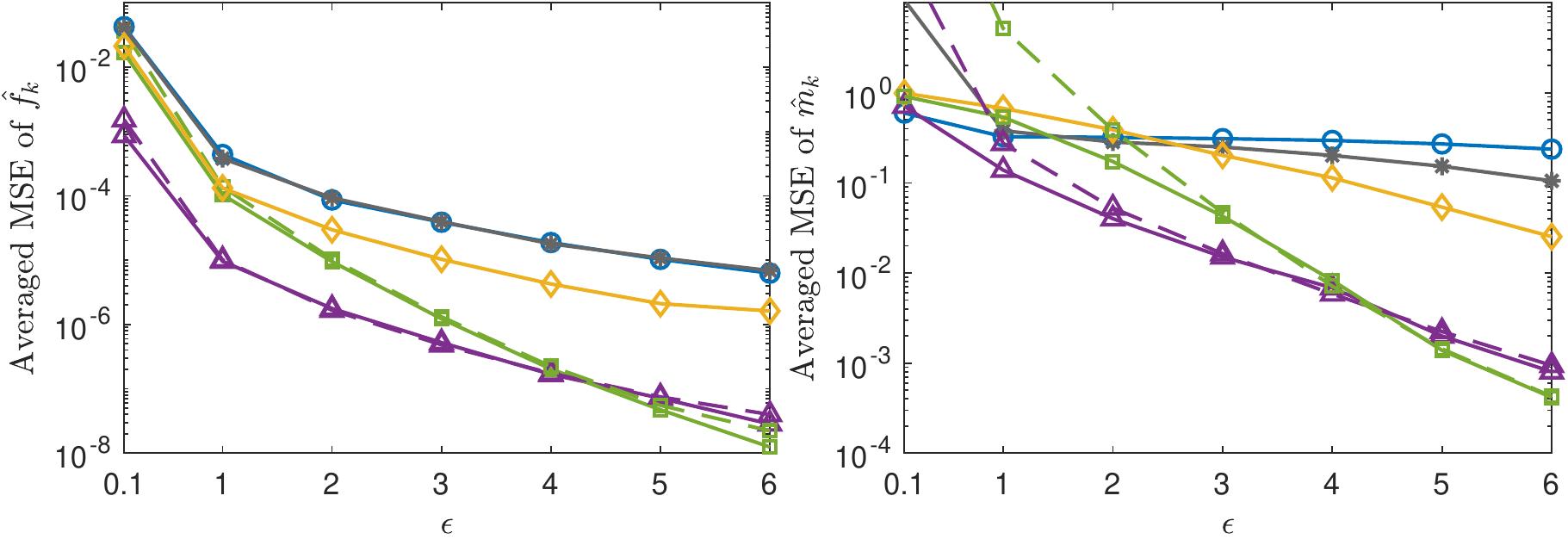}
    (b) Gaussian distribution (MSE is averaged of all keys)
    \includegraphics[width=3.2in]{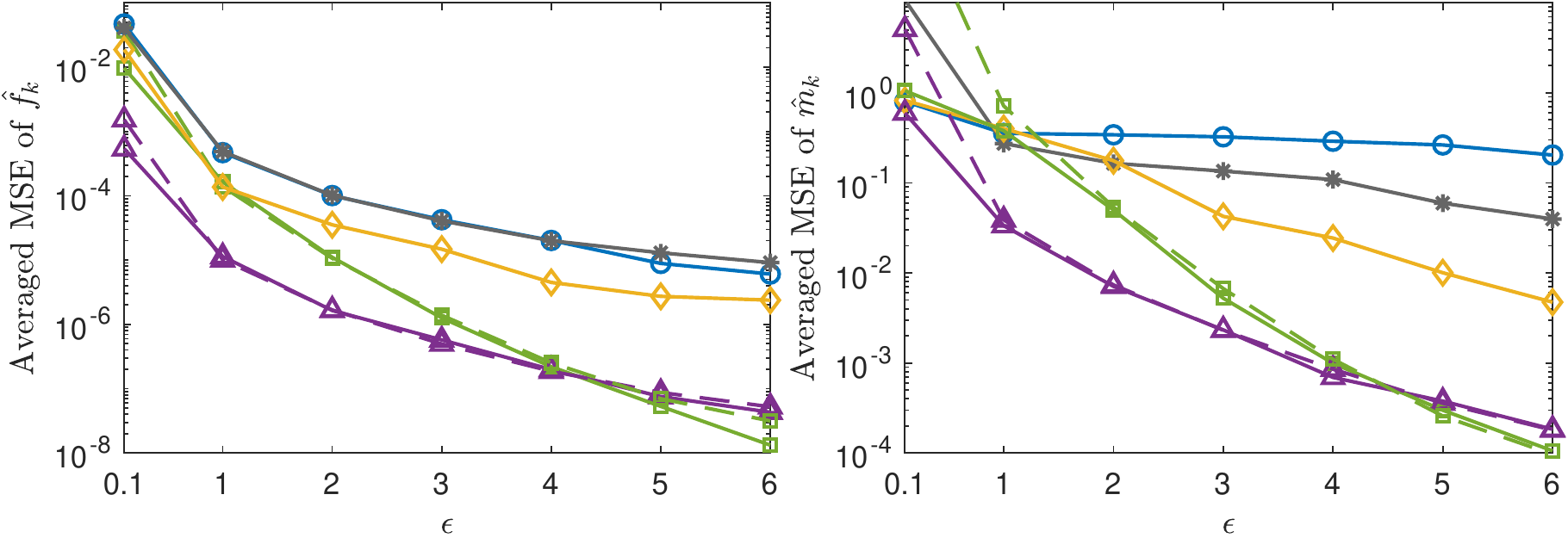}
    (c) Gaussian distribution (MSE is averaged of top 20 frequent keys)
    \\\vspace{-2mm}
    \caption{\small MSEs of synthetic data under two distributions, where the left is MSE of frequency estimation and the right is MSE of mean estimation. The theoretical MSEs (dashed lines) of PCKV-UE and PCKV-GRR are calculated by Theorem \ref{thm:estimation}. When $\epsilon$ is small, the gap between empirical and theoretical results is caused by the correction in the improved estimation (Algorithm \ref{alg:Estimation}), while our theoretical MSE is analyzed for the baseline estimation without correction.}
    \label{fig:synthetic}
\end{figure}

\textbf{Overall Results.}
The averaged MSEs of frequency and mean estimations are shown in Figure \ref{fig:synthetic} (with domain size 100), where the MSE is averaged by all keys (Figure \ref{fig:synthetic}a and \ref{fig:synthetic}b) or the top 20 frequent keys (Figure \ref{fig:synthetic}c). For frequency estimation, PrivKVM ($c=1$) and PrivKVM (\verb|1r5v|) have the same MSE since the frequency is estimated by the first iteration. The proposed mechanisms (PCKV-UE and PCKV-GRR) have much less $\text{MSE}_{\hat{f}_k}$. For mean estimation, PrivKVM (\verb|1r5v|) predicts the mean estimation of remaining iterations without splitting the budget, which improves the accuracy of PrivKVM ($c=1$) under  larger $\epsilon$. The $\text{MSE}_{\hat{m}_k}$ of PrivKVM ($c=1$) does not decrease any more after $\epsilon=0.5$ since PrivKVM ($c=1$) always generates fake values as $v=0$. The PrivKVM (\verb|1r5v|) with virtual iterations improves PrivKVM ($c=1$), but the estimation error is larger than other mechanisms. The $\text{MSE}_{\hat{m}_k}$ in PCKV-UE and PCKV-GRR is much smaller than other ones when $\epsilon$ is relatively large (e.g., $\epsilon>2$), thanks to the high accuracy of frequency estimation in this case. Also, the small gap between the theoretical and empirical results validate the correctness of our theoretical error analysis in Theorem \ref{thm:estimation}.

\textbf{Influence of Data Distribution.} By comparing the results of PCKV-UE and PCKV-GRR under different distributions in Figure \ref{fig:synthetic},  $\text{MSE}_{\hat{m}_k}$ of all keys in Gaussian distribution is larger than  in Uniform distribution because the frequency of some keys is very small in Gaussian distribution. However,  $\text{MSE}_{\hat{m}_k}$ of the top 20 frequent keys is much smaller because the frequent keys have higher frequencies. Note that the distribution has little influence on $\text{MSE}_{\hat{f}_k}$ in these mechanisms because the user size used in frequency estimation is always $n$, while the user size used in value mean estimation of $k\in\mathcal{K}$ is $nf_k^{*}$.

\begin{figure}[!t]
    \centering
    \footnotesize
    \includegraphics[width=3.2in]{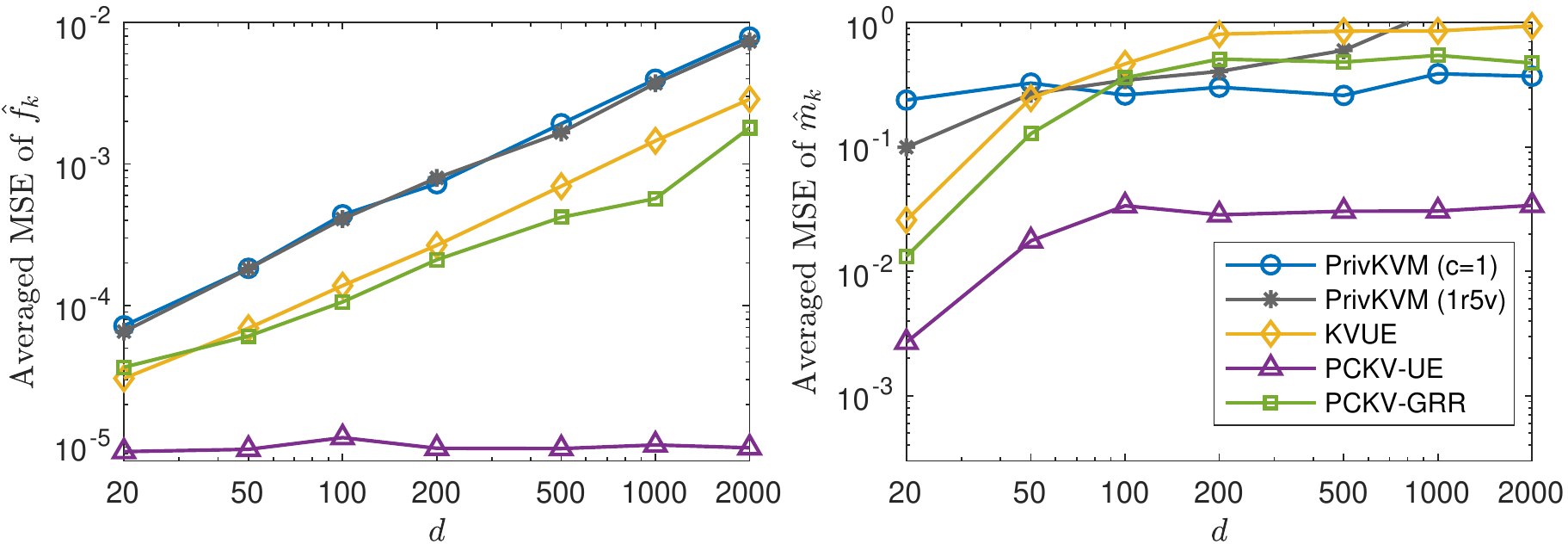}
    (a) Gaussian distribution (with $\epsilon=1$)
    \includegraphics[width=3.2in]{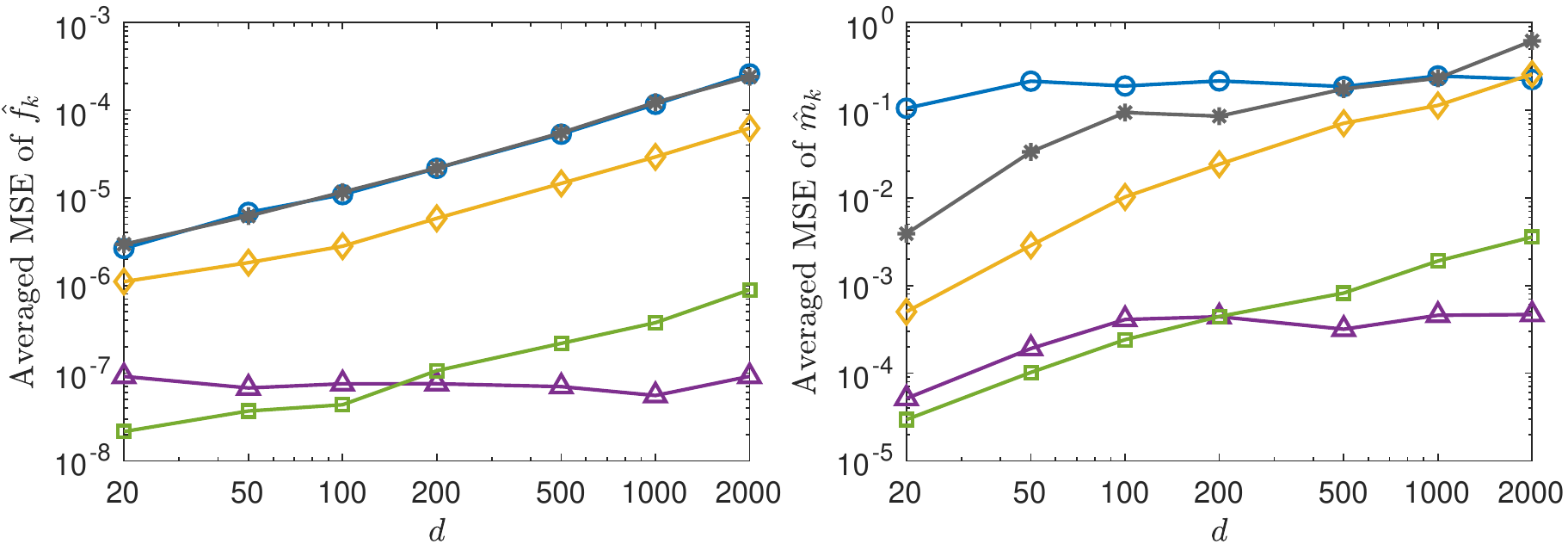}
    (b) Gaussian distribution (with $\epsilon=5$)
    \vspace{-3mm}
    \caption{Varying domain size $d$ (MSEs are averaged of the top 20 frequent keys). }
    \vspace{-4mm}
    \label{fig:domain_size}
\end{figure}

\begin{figure}[!t]
    \centering
    \includegraphics[width=3.3in]{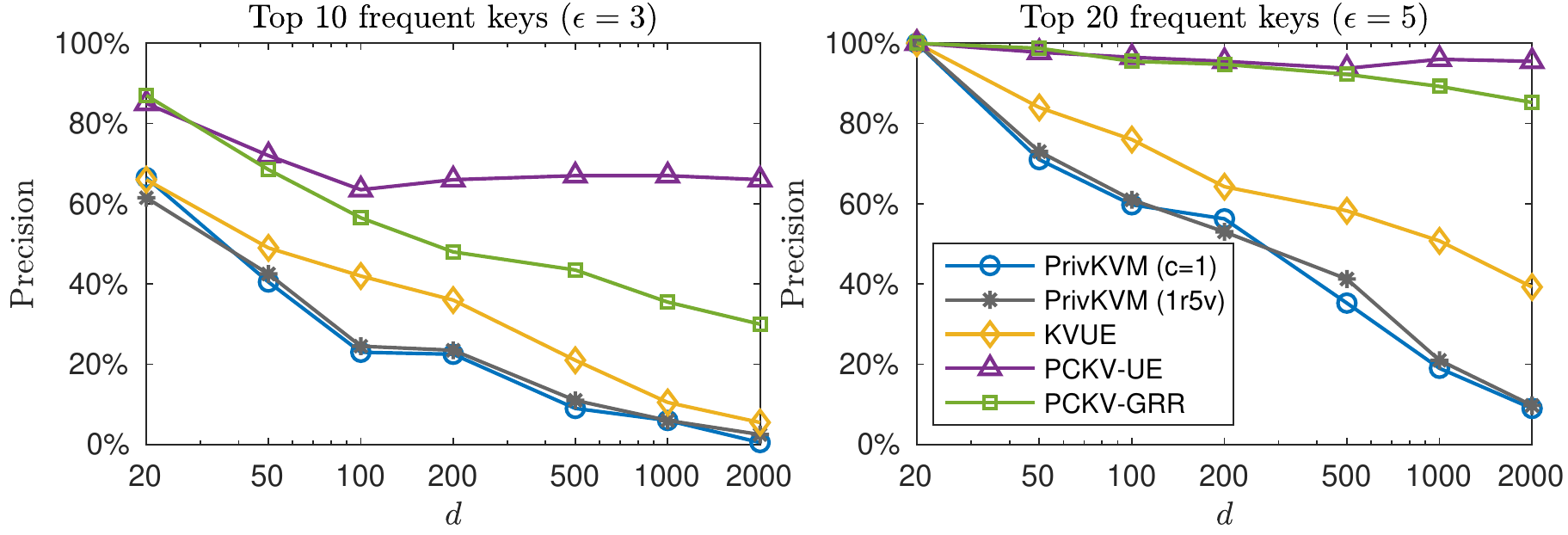}
    \vspace{-4mm}
    \caption{Precision of top frequent keys estimation.}
    \label{fig:success_top}
\end{figure}

\textbf{Influence of Domain Size.}
The MSEs of frequency and mean estimation with respect to different domain size $d$ (where $\epsilon=1$ or $5$) are shown in Figure \ref{fig:domain_size}. We can observe that $\text{MSE}_{\hat{f}_k}$ is proportional to the domain size $d$ in PrivKVM, KVUE, and PCKV-GRR. Note that the reasons for the same observation are different. For PrivKVM and KVUE, the perturbation probabilities are independent of domain size, but the large domain size would make sampling protocol (randomly pick one index from the domain of keys) less possible to obtain the useful information.  For PCKV-GRR, the large domain size does not influence the Padding-and-Sampling protocol, but it will decrease the perturbation probabilities $a$ and $b$ in \eqref{equ:abp_GRR_opt} and enlarge the estimation error. However, the large domain size does not affect the frequency estimation of PCKV-UE. For the result of mean estimation, we have similar observations.  Note that $\text{MSE}_{\hat{m}_k}$ is not proportional to the domain size because the correction of mean estimation can alleviate the error. For PCKV-UE, the increasing $\text{MSE}_{\hat{m}_k}$ when $d<100$ is caused by the decreased true frequency when $d$ is increasing (note that $\sigma_{\text{key}}=50$ and samples outside the domain are discarded when generating the data). The prediction of PrivKVM (\verb|1r5v|) with virtual iterations does not work well for a large domain size under small $\epsilon$.

\textbf{Accuracy of Top Frequent Keys Selection.} To evaluate the success of the top frequent keys selection, we calculate the precision (i.e., the proportion of correct selections over all predicted top frequent keys) for different mechanisms, shown in Figure \ref{fig:success_top} (precision in this case is the same as recall). For the top 10 frequent keys under $\epsilon=3$,  the precision of PCKV-UE is over $60\%$ even for a large $d$ (i.e., misestimation is at most 4 over the top 10 frequent keys). However, PrivKVM and KVUE incorrectly select almost all top 10 frequent keys when $d=2000$. For the top 20 frequent keys under $\epsilon=5$, PCKV-UE and PCKV-GRR can correctly estimate $95\%$ and $85\%$ respectively even for $d=2000$.

\textbf{Comparison of Allocation Methods.}
In our PCKV framework, the privacy-utility tradeoff is improved by both the tighter bound in budget composition (Theorem \ref{thm:LDP_UE} and Theorem \ref{thm:LDP_GRR}) and the optimized budget allocation (Lemma \ref{lem:budget_UE} and Lemma \ref{lem:budget_GRR}). In order to show the benefit of our optimized allocation, we compare the results of optimized method with two alternative allocation ones in Figure \ref{fig:PCKV_allocation}, where the corresponding theoretical comparison has been discussed in Sec. \ref{sec:allocation}. The naive allocation is $\epsilon_1=\epsilon_2=\epsilon/2$, and the non-optimized allocation with tighter bound is represented in (\ref{equ:non-optimized}), which also works for PCKV-GRR when $\ell=1$.  We can observe that for both PCKV-UE and PCKV-GRR, the allocation methods with tighter bound (non-optimized and optimized) outperform the naive one in the estimation accuracy of mean and frequency. Even though $\text{MSE}_{\hat{f}_k}$ in optimized allocation is slightly greater than the  non-optimized one, it has much less  $\text{MSE}_{\hat{m}_k}$.  Note that the magnitude of $\text{MSE}_{\hat{f}_k}$ and $\text{MSE}_{\hat{m}_k}$ are different. For example, when $\epsilon=1$, the gap of $\text{MSE}_{\hat{f}_k}$ between non-optimized and optimized allocation in PCKV-UE is $4\times 10^{-6}$, but the gap of $\text{MSE}_{\hat{m}_k}$ between them is 0.08. These observations validate our theoretical analyses and discussions in Sec. \ref{sec:allocation}.

\begin{figure}[!t]
    \centering
    \includegraphics[width=3.3in]{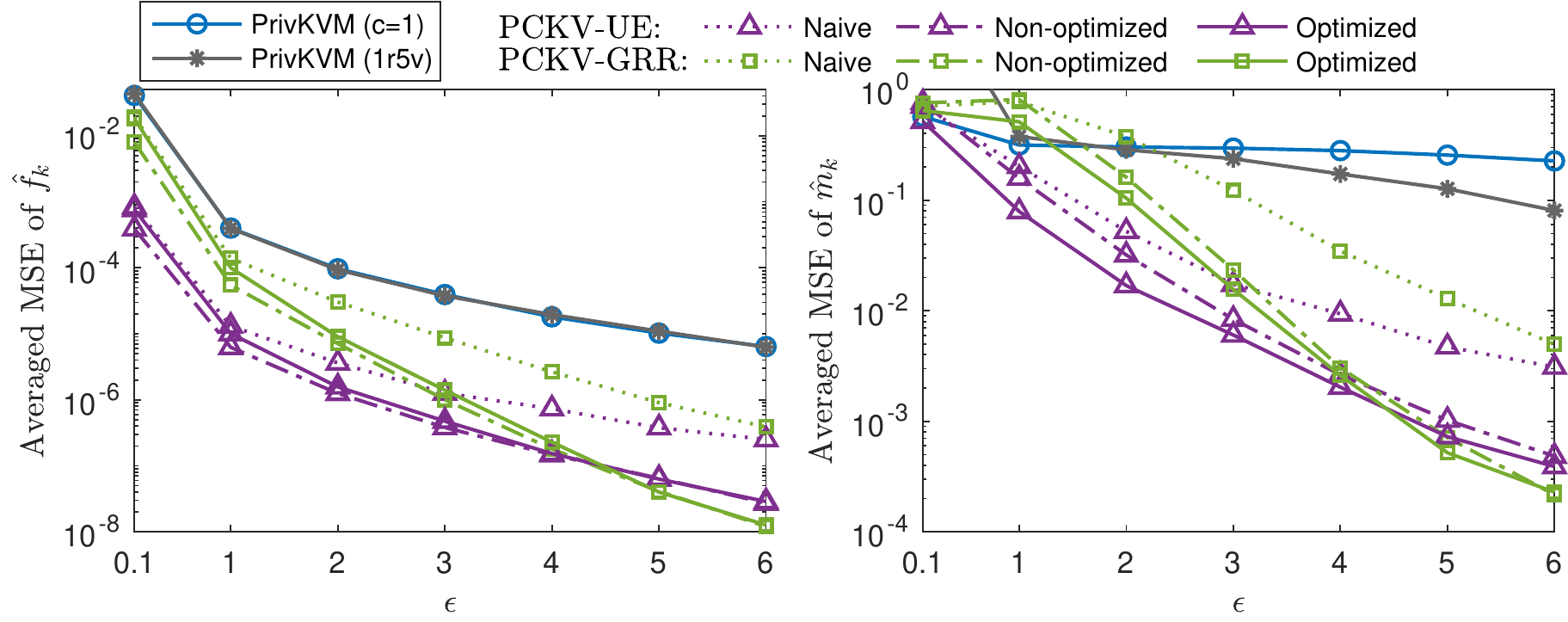}
    \vspace{-4mm}
    \caption{Comparison of three allocation methods in PCKV.}
    \label{fig:PCKV_allocation}
\end{figure}

\begin{figure}[!t]
    \centering
    \subfloat[E-commerce dataset with $n=23,486$, $d=1,206$ and $\ell=1$.]{\includegraphics[width=3.3in]{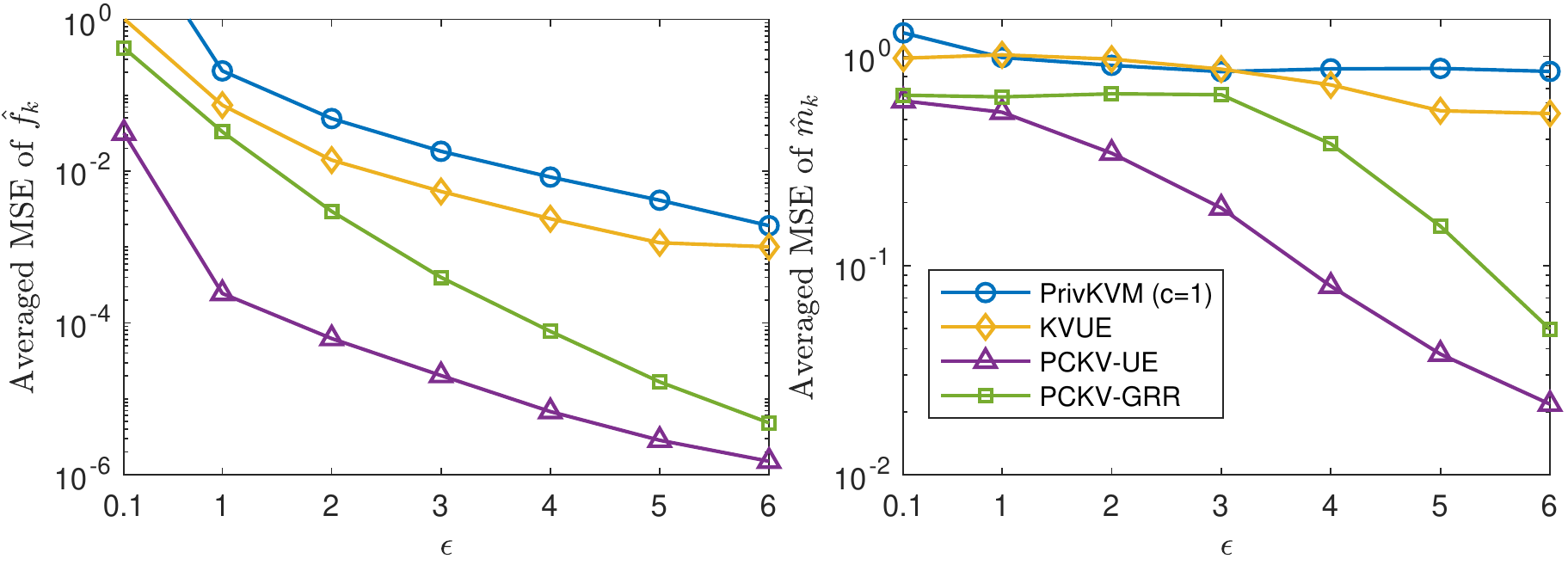}}\\
    \vspace{-3mm}
    \subfloat[Clothing dataset with $n=105,508$, $d=5,850$ and $\ell=2$.]{\includegraphics[width=3.3in]{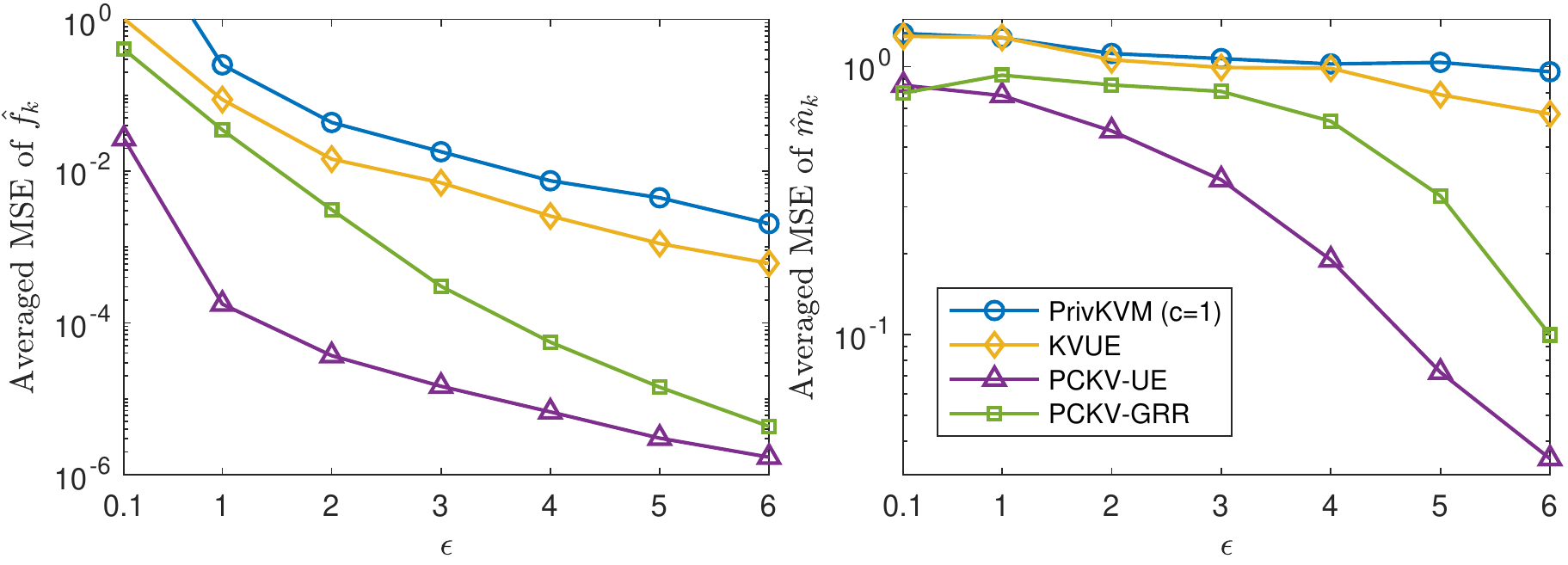}}\\
    \vspace{-3mm}
    \subfloat[Amazon dataset with $n=1,210,271$, $d=229,274$ and $\ell=2$.]{\includegraphics[width=3.3in]{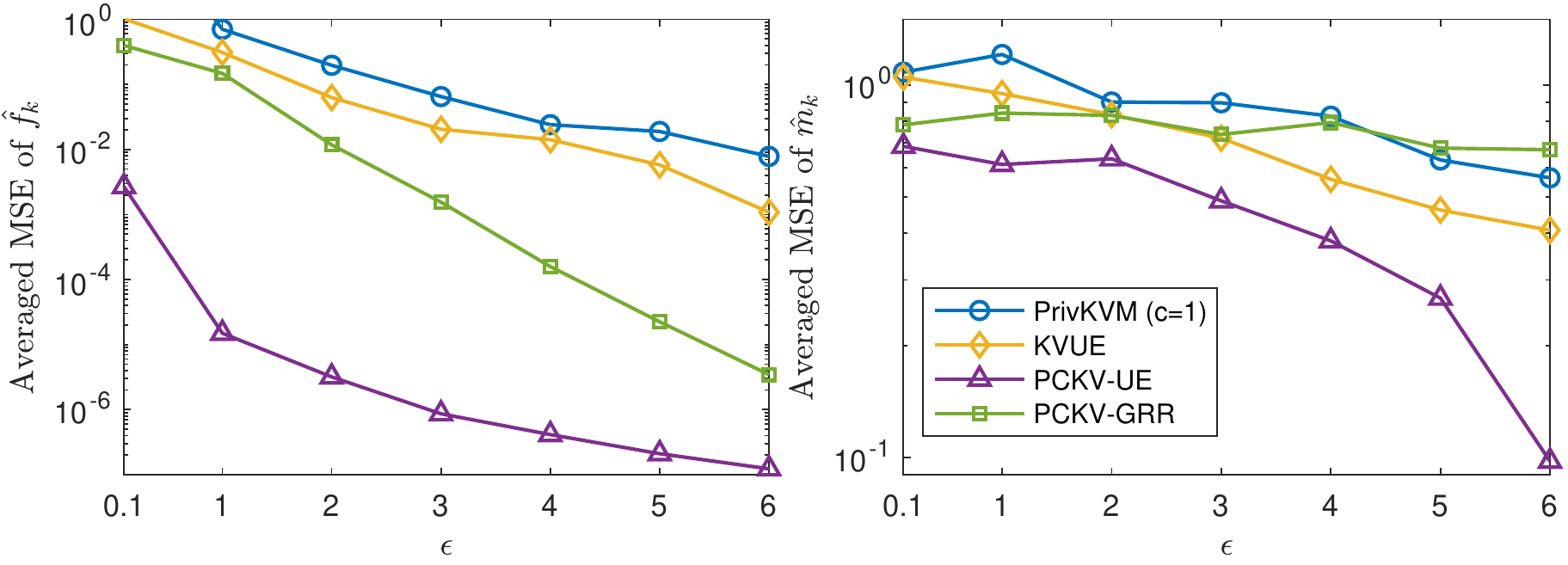}}\\
    \vspace{-3mm}
    \subfloat[Movie dataset with $n=138,493$, $d=26,744$ and $\ell=100$.]{\includegraphics[width=3.3in]{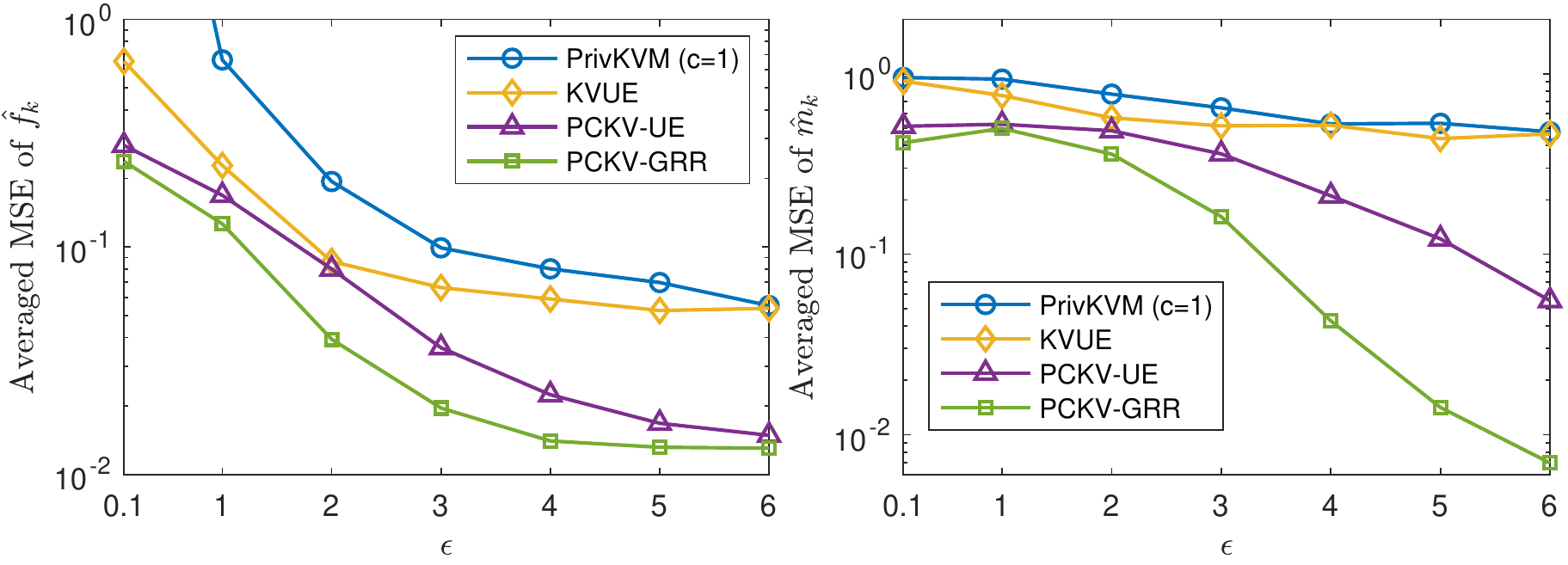}}
    \vspace{-3mm}
    \caption{MSEs of real-world datasets listed in Table \ref{tab:data}.}
    \label{fig:real-world}
\end{figure}

\subsection{Real-World Data}

The results of four types of real-world rating datasets are shown in Figure \ref{fig:real-world}, where the MSEs are averaged over the top 50 frequent keys. The parameters (number of ratings, users, and keys)  are listed in Table \ref{tab:data}, where we select reasonable $\ell$ for evaluation to compare with existing mechanisms with naive sampling protocol (the advanced strategy of selecting an optimized $\ell$ is discussed in Sec. \ref{sec:sampling}). Under the large domain size in real-world datasets, PrivKVM (\verb|1r5v|) with virtual iterations does not work well, thus we only show the results of PrivKVM ($c=1$). Compared with the results of E-commerce dataset, the MSEs of Clothing dataset do not change very much because all algorithms can get benefits from the large $n$, which compensates the impacts from the larger $d$ or the larger $\ell$. Compared with the results of PCKV-UE in Clothing dataset, $\text{MSE}_{\hat{f}_k}$ in Amazon dataset is smaller (due to the large $n$) but $\text{MSE}_{\hat{m}_k}$ is larger (due to the small true frequencies). In the first three datasets, PCKV-UE has the best performance because $\ell$ is small and the large domain size does not impact its performance directly. In the Movie dataset, since PCKV-GRR can benefit more from a large $\ell$, it outperforms PCKV-UE in both frequency and mean estimation. Note that both PCKV-UE and PCKV-GRR have less MSEs compared with other mechanisms in Movie dataset. Since PCKV-UE and PCKV-GRR are suitable for different cases, in practice we can select PCKV-UE or PCKV-GRR by comparing the theoretical estimation error under specified parameters (i.e., $\epsilon,d$ and $\ell$) as discussed in Sec. \ref{sec:allocation}.

\section{Conclusion}
In this paper, a new framework called PCKV (with two mechanisms PCKV-UE and PCKV-GRR) is proposed to privately collect key-value data under LDP with higher accuracy of frequency and value mean estimation. We design a correlated key and value perturbation mechanism that leads to a tighter budget composition than sequential composition of LDP. We further improve the privacy-utility tradeoff via a near-optimal budget allocation method. Besides the tight budget composition and optimized budget allocation, the proposed sampling protocol and mean estimators in our framework also improve the accuracy of estimation than existing protocols. Finally, we demonstrate the advantage of the proposed scheme on both synthetic and real-world datasets. 

For future work, we will study how to choose an optimized $\ell$ in the Padding-and-Sampling protocol  and  extend the correlated perturbation and tight composition analysis to consider more general forms of correlation and other hybrid data types.

\section*{Acknowledgments}
Yueqiang Cheng is the corresponding author (main work was done when the first author was a summer intern at Baidu X-Lab). The authors would like to thank the anonymous reviewers and the shepherd Mathias L{\'e}cuyer for their valuable comments and suggestions. This research was partially sponsored by NSF grants CNS-1731164 and CNS-1618932,  JSPS grant KAKENHI-19K20269, AFOSR grant FA9550-12-1-0240, and NIH grant R01GM118609.

\bibliographystyle{plain}
\bibliography{mybibfile}

\begin{thebibliography}{10}

\bibitem{amazon}
Amazon rating dataset.
\newblock \url{https://www.kaggle.com/skillsmuggler/amazon-ratings}.

\bibitem{clothing}
Clothing fit and rating dataset.
\newblock
  \url{https://www.kaggle.com/rmisra/clothing-fit-dataset-for-size-recommendation}.

\bibitem{ecommerce}
Ecommerce rating dataset.
\newblock
  \url{https://www.kaggle.com/nicapotato/womens-ecommerce-clothing-reviews}.

\bibitem{movie}
Movie rating dataset.
\newblock \url{https://www.kaggle.com/ashukr/movie-rating-data}.

\bibitem{apple2017learning}
Learning with privacy at scale.
\newblock
  \url{https://machinelearning.apple.com/2017/12/06/learning-with-privacy-at-scale.html},
  2017.

\bibitem{bassily2015local}
Raef Bassily and Adam Smith.
\newblock Local, private, efficient protocols for succinct histograms.
\newblock In {\em ACM Symposium on Theory of Computing (STOC)}, pages 127--135,
  2015.

\bibitem{casella2002statistical}
George Casella and Roger~L Berger.
\newblock {\em Statistical inference}.
\newblock Duxbury Pacific Grove, CA, 2002.

\bibitem{chen2016private}
Rui Chen, Haoran Li, AK~Qin, Shiva~Prasad Kasiviswanathan, and Hongxia Jin.
\newblock Private spatial data aggregation in the local setting.
\newblock In {\em IEEE International Conference on Data Engineering}, pages
  289--300, 2016.

\bibitem{ding2017collecting}
Bolin Ding, Janardhan Kulkarni, and Sergey Yekhanin.
\newblock Collecting telemetry data privately.
\newblock In {\em Advances in Neural Information Processing Systems}, pages
  3571--3580, 2017.

\bibitem{duchi2013local}
John~C Duchi, Michael~I Jordan, and Martin~J Wainwright.
\newblock Local privacy and statistical minimax rates.
\newblock In {\em IEEE Symposium on Foundations of Computer Science (FOCS)},
  pages 429--438, 2013.

\bibitem{duchi2018minimax}
John~C Duchi, Michael~I Jordan, and Martin~J Wainwright.
\newblock Minimax optimal procedures for locally private estimation.
\newblock {\em Journal of the American Statistical Association},
  113(521):182--201, 2018.

\bibitem{dwork2006differential}
Cynthia Dwork.
\newblock Differential privacy.
\newblock In {\em ICALP}, pages 1--12, 2006.

\bibitem{dwork2006calibrating}
Cynthia Dwork, Frank McSherry, Kobbi Nissim, and Adam Smith.
\newblock Calibrating noise to sensitivity in private data analysis.
\newblock In {\em Theory of Cryptography Conference (TCC)}, pages 265--284,
  2006.

\bibitem{erlingsson2014rappor}
{\'U}lfar Erlingsson, Vasyl Pihur, and Aleksandra Korolova.
\newblock Rappor: Randomized aggregatable privacy-preserving ordinal response.
\newblock In {\em ACM Conference on Computer and Communications Security},
  pages 1054--1067, 2014.

\bibitem{mcsherry2009privacy}
Frank~D McSherry.
\newblock Privacy integrated queries: an extensible platform for
  privacy-preserving data analysis.
\newblock In {\em ACM SIGMOD International Conference on Management of data},
  pages 19--30, 2009.

\bibitem{nguyen2016collecting}
Th{\^o}ng~T Nguy{\^e}n, Xiaokui Xiao, Yin Yang, Siu~Cheung Hui, Hyejin Shin,
  and Junbum Shin.
\newblock Collecting and analyzing data from smart device users with local
  differential privacy.
\newblock {\em arXiv preprint: 1606.05053}, 2016.

\bibitem{qin2016heavy}
Zhan Qin, Yin Yang, Ting Yu, Issa Khalil, Xiaokui Xiao, and Kui Ren.
\newblock Heavy hitter estimation over set-valued data with local differential
  privacy.
\newblock In {\em ACM SIGSAC Conference on Computer and Communications Security
  (CCS)}, pages 192--203, 2016.

\bibitem{qin2017generating}
Zhan Qin, Ting Yu, Yin Yang, Issa Khalil, Xiaokui Xiao, and Kui Ren.
\newblock Generating synthetic decentralized social graphs with local
  differential privacy.
\newblock In {\em ACM SIGSAC Conference on Computer and Communications Security
  (CCS)}, pages 425--438, 2017.

\bibitem{ren2018textsf}
Xuebin Ren, Chia-Mu Yu, Weiren Yu, Shusen Yang, Xinyu Yang, Julie~A McCann, and
  S~Yu Philip.
\newblock Lopub: High-dimensional crowdsourced data publication with local
  differential privacy.
\newblock {\em IEEE Transactions on Information Forensics and Security},
  13(9):2151--2166, 2018.

\bibitem{sun2019conditional}
Lin Sun, Jun Zhao, Xiaojun Ye, Shuo Feng, Teng Wang, and Tao Bai.
\newblock Conditional analysis for key-value data with local differential
  privacy.
\newblock {\em arXiv preprint arXiv:1907.05014}, 2019.

\bibitem{wang2019collecting}
Ning Wang, Xiaokui Xiao, Yin Yang, Jun Zhao, Siu~Cheung Hui, Hyejin Shin,
  Junbum Shin, and Ge~Yu.
\newblock Collecting and analyzing multidimensional data with local
  differential privacy.
\newblock In {\em IEEE International Conference on Data Engineering (ICDE)},
  pages 1--12, 2019.

\bibitem{wang2017locally}
Tianhao Wang, Jeremiah Blocki, Ninghui Li, and Somesh Jha.
\newblock Locally differentially private protocols for frequency estimation.
\newblock In {\em USENIX Security Symposium}, pages 729--745, 2017.

\bibitem{wang2018locally}
Tianhao Wang, Ninghui Li, and Somesh Jha.
\newblock Locally differentially private frequent itemset mining.
\newblock In {\em IEEE Symposium on Security and Privacy (S\&P)}, 2018.

\bibitem{warner1965randomized}
Stanley~L Warner.
\newblock Randomized response: A survey technique for eliminating evasive
  answer bias.
\newblock {\em Journal of the American Statistical Association}, 1965.

\bibitem{ye2019privkv}
Qingqing Ye, Haibo Hu, Xiaofeng Meng, and Huadi Zheng.
\newblock Privkv: Key-value data collection with local differential privacy.
\newblock In {\em IEEE Symposium on Security and Privacy (S\&P)}, 2019.

\bibitem{zhang2018calm}
Zhikun Zhang, Tianhao Wang, Ninghui Li, Shibo He, and Jiming Chen.
\newblock Calm: Consistent adaptive local marginal for marginal release under
  local differential privacy.
\newblock In {\em ACM SIGSAC Conference on Computer and Communications Security
  (CCS)}, pages 212--229, 2018.

\end{thebibliography}

\appendix

\section{Proof of Theorem \ref{thm:LDP_UE}}
\label{apx:thm:LDP_UE}
\begin{proof}
For a key-value set $\mathcal{S}$, denote the key-value pairs (raw data) are $\langle i,v_i^{*}\rangle$ for all $i\in\mathcal{S}$, where $v_i^{*}\in[-1,1]$. Note that $i\in\mathcal{S}$ means a key-value pair $\langle i,\cdot\rangle\in\mathcal{S}$. Denote the sampled key-value pair by Padding-and-Sampling in Algorithm \ref{alg:PS} as $x=\langle k,v\rangle$, where $v\in\{1,-1\}$ (the discretized value). According to Line-5 in Algorithm \ref{alg:PS}, we have $v_k^{*}=0$ for $k\in\{d+1,\cdots,d^{\prime}\}$, where $d^{\prime}=d+\ell$. For vector $\mathbf{x}$ in PCKV-UE, only the $k$-th element is $v$ ($1$ or $-1$) while others are 0s. Then, the probability of outputting a vector $\mathbf{y}$  is
\begin{align*}
    \Pr(\mathbf{y}|\mathcal{S},k)
    &=\Pr(\mathbf{y}[k]|v_k^{*})\prod_{i\in\mathcal{K}^{\prime}\backslash k}\Pr(\mathbf{y}[i]|\mathbf{x}[i]=0)\\
    &=\frac{\Pr(\mathbf{y}[k]|v_k^{*})}{\Pr(\mathbf{y}[k]|\mathbf{x}[k]=0)}\cdot\prod_{i\in\mathcal{K}^{\prime}}\Pr(\mathbf{y}[i]|\mathbf{x}[i]=0)
\end{align*}
According to Figure \ref{fig:perturbation}, 
the perturbation probabilities of the $k$-th element from the raw value can be represented as
\begin{align*}
\Pr(\mathbf{y}[k]|v_k^{*})=
\begin{cases}
\frac{1+(2p-1)v_k^{*}}{2}\cdot a,& \text{if }\mathbf{y}[k]=1\\
\frac{1-(2p-1)v_k^{*}}{2}\cdot a,& \text{if }\mathbf{y}[k]=-1\\
1-a,& \text{if }\mathbf{y}[k]=0\\
\end{cases}
\end{align*}
where $v_k^{*}\in[-1,1]$. For convenience, denote 
\begin{align*}
    \Psi(\mathbf{y},k) =\frac{\Pr(\mathbf{y}[k]|v_k^{*})}{\Pr(\mathbf{y}[k]|\mathbf{x}[k]=0)},~~
    \Phi(\mathbf{y}) = \prod_{i\in\mathcal{K}^{\prime}}\Pr(\mathbf{y}[i]|\mathbf{x}[i]=0)
\end{align*}
then we have $\Pr(\mathbf{y}|\mathcal{S},k)=\Psi(\mathbf{y},k)\cdot\Phi(\mathbf{y})$ and
\begin{align*}
\Psi(\mathbf{y},k)=
\begin{cases}
(1+(2p-1)v_k^{*})\cdot\frac{a}{b},& \text{if }\mathbf{y}[k]=1\\
(1-(2p-1)v_k^{*})\cdot\frac{a}{b},& \text{if }\mathbf{y}[k]=-1\\
\frac{1-a}{1-b},& \text{if }\mathbf{y}[k]=0\\
\end{cases}
\end{align*}
where $a,p\in[\frac{1}{2},1)$ and $b\in(0,\frac{1}{2}]$ (in Algorithm \ref{alg:PCKV-UE}). 

\textbf{Case 1.}
For $k\in\{1,2,\cdots,d\}$, we have $v_k^{*}\in[-1,1]$ and
\begin{align*}
    \frac{1-a}{1-b}\leqslant\frac{2pa}{b},\quad
    \frac{2(1-p)a}{b}\leqslant (1\pm(2p-1)v_k^{*})\cdot\frac{a}{b}\leqslant
    \frac{2pa}{b}
\end{align*}
then the upper bound and lower bound of $\Psi(\mathbf{y},k)$ are
\begin{align*}
   \Psi_\text{upper}=\frac{2pa}{b},\quad 
   \Psi_\text{lower}=\min\left\{\frac{1-a}{1-b},\frac{2(1-p)a}{b}\right\}
\end{align*}

\textbf{Case 2.}
For $k\in\{d+1,\cdots,d^{\prime}\}$, we have $v_k^{*}=0$, then the upper bound and lower bound of $\Psi(\mathbf{y},k)$  are
\begin{align*}
   \Psi_\text{upper}^{\prime}=\frac{a}{b},\quad 
   \Psi_\text{lower}^{\prime}=\frac{1-a}{1-b}
\end{align*}

Note that $\Psi_\text{lower}\leqslant\Psi_\text{lower}^{\prime}\leqslant\Psi_\text{upper}^{\prime}\leqslant\Psi_\text{upper}$. Then, the probability of perturbing $\mathcal{S}$ into $\mathbf{y}$ is bounded by
\begin{align*}
    &\Pr(\mathbf{y}|\mathcal{S})
    = \eta\sum_{k\in\mathcal{S}}\frac{\Pr(\mathbf{y}|\mathcal{S},k)}{|\mathcal{S}|} + (1-\eta)\sum_{k=d+1}^{d^{\prime}}\frac{\Pr(\mathbf{y}|\mathcal{S},k)}{\ell}\\
    &=\Phi(\mathbf{y})\left[\frac{\eta}{|\mathcal{S}|}\sum_{k\in\mathcal{S}}\Psi(\mathbf{y},k)+\frac{1-\eta}{\ell}\sum_{k=d+1}^{d^{\prime}}\Psi(\mathbf{y},k)\right]\\
    &\leqslant \Phi(\mathbf{y})\left[\frac{\eta}{|\mathcal{S}|}\cdot |\mathcal{S}|\Psi_\text{upper} +\frac{1-\eta}{\ell}\cdot \ell\Psi_\text{upper}^{\prime}\right]
    \leqslant\Phi(\mathbf{y})\cdot\Psi_\text{upper}
\end{align*}
where the last inequality holds since $\eta=\frac{|\mathcal{S}|}{\max\{|\mathcal{S}|,\ell\}}\in(0,1]$ and $\Psi_\text{upper}^{\prime}\leqslant\Psi_\text{upper}$. Similarly, $\Pr(\mathbf{y}|\mathcal{S}) \geqslant \Phi(\mathbf{y})\cdot\Psi_\text{lower}$ holds. Then, for two different key-value sets $\mathcal{S}_1$ and $\mathcal{S}_2$, we have
\begin{align*}
    &\frac{\Pr(\mathbf{y}|\mathcal{S}_1)}{\Pr(\mathbf{y}|\mathcal{S}_2)}
    \leqslant\frac{\Phi(\mathbf{y})\cdot\Psi_\text{upper}}{\Phi(\mathbf{y})\cdot\Psi_\text{lower}}
    =\frac{\Psi_\text{upper}}{\Psi_\text{lower}}
    =\frac{2pa/b}{\min\left\{\frac{1-a}{1-b},\frac{2(1-p)a}{b}\right\}}\\
    &=\max\left\{2p\cdot\frac{a(1-b)}{b(1-a)},\frac{p}{1-p}\right\}
    = \max\left\{\frac{2e^{\epsilon_1}}{1+e^{-\epsilon_2}},e^{\epsilon_2}\right\}=e^{\epsilon}
\end{align*}
where $\epsilon$ is defined in (\ref{equ:epsilon_UE}).
\end{proof}

\section{Proof of Theorem \ref{thm:LDP_GRR}}
\label{apx:thm:LDP_GRR}
\begin{proof}
In PCKV-GRR, for an input $\mathcal{S}$ with pairs $\langle i,v_i^{*}\rangle$ for all $i\in\mathcal{S}$ and an output $y^{\prime}=\langle k^{\prime},v^{\prime}\rangle$ , denote the sampled pair as $x=\langle k,v\rangle$.  When the sampled key is $k$, the probability of outputting a pair $y^{\prime}=\langle k^{\prime},v^{\prime}\rangle$ is
\begin{align*}
    \Pr(y^{\prime}|\mathcal{S},k)=
   \begin{cases}
   \frac{1+(2p-1)v_k^{*}}{2}\cdot a,& \text{if } k^{\prime}=k, v^{\prime}=1\\
   \frac{1-(2p-1)v_k^{*}}{2}\cdot a,& \text{if } k^{\prime}=k, v^{\prime}=-1\\
   b/2,& \text{if } k^{\prime}\neq k
   \end{cases}
\end{align*}
where $v_k^{*}=0$ for $k\in\{d+1,\cdots,d^{\prime}\}$.

\textbf{Case 1.} If $k^{\prime}\in\mathcal{S}$, then
\begin{align*}
    & \Pr(y^{\prime}|\mathcal{S})=\eta\sum_{k\in\mathcal{S}}\frac{\Pr(y^{\prime}|\mathcal{S},k)}{|\mathcal{S}|} + (1-\eta)\sum_{k=d+1}^{d^{\prime}}\frac{\Pr(y^{\prime}|\mathcal{S},k)}{\ell}\\
    &= \frac{\eta}{|\mathcal{S}|}\left[a\cdot\frac{1+(2p-1)v_{k^{\prime}}^{*}v^{\prime}}{2}+(|\mathcal{S}|-1)\frac{b}{2}\right]+(1-\eta)\frac{b}{2}
\end{align*}
Considering $v_{k^{\prime}}^{*}\in[-1,1]$ and $v^{\prime}\in\{1,-1\}$, we have
\begin{align}
    \label{equ:case1_upper}
    \Pr(y^{\prime}|\mathcal{S})\leqslant\frac{\eta}{|\mathcal{S}|}ap+(1-\frac{\eta}{|\mathcal{S}|})\frac{b}{2}\leqslant
    \frac{1}{\ell}ap+(1-\frac{1}{\ell})\frac{b}{2}
\end{align}
where $\frac{\eta}{|\mathcal{S}|}=\frac{1}{\max\{|\mathcal{S}|,\ell\}}\in[\frac{1}{d},\frac{1}{\ell}]$ and $ap>\frac{1}{4}>\frac{b}{2}$. Also, 
\begin{align}
    \label{equ:case1_lower}
    \Pr(y^{\prime}|\mathcal{S})\geqslant\frac{\eta}{|\mathcal{S}|}a(1-p)+(1-\frac{\eta}{|\mathcal{S}|})\frac{b}{2}
\end{align}

\textbf{Case 2.} If $k^{\prime}\notin\mathcal{S}$, i.e., $k^{\prime}\in\{d+1,\cdots,d^{\prime}\}$, then
\begin{align}
    \label{equ:case2_upper}
    &\quad\Pr(y^{\prime}|\mathcal{S})=\eta\cdot\frac{b}{2}+\frac{1-\eta}{\ell}\left[\frac{a}{2}+(\ell-1)\frac{b}{2}\right] \notag\\
    &<\frac{1}{\ell}\left[\frac{a}{2}+(\ell-1)\frac{b}{2}\right]
    <\frac{1}{\ell}ap+(1-\frac{1}{\ell})\frac{b}{2}
\end{align}
where $\eta=\frac{|\mathcal{S}|}{\max\{|\mathcal{S}|,\ell\}}\in[\frac{1}{\ell},1]$, and $a,p>\frac{1}{2}>b$. Also,
\begin{align}
    \label{equ:case2_lower}
    \Pr(y^{\prime}|\mathcal{S})=\eta\cdot\frac{b}{2}+\frac{1-\eta}{\ell}\left[\frac{a}{2}+(\ell-1)\frac{b}{2}\right] \geqslant\frac{b}{2}
\end{align}

\textbf{Bound of Probability Ratio.} Denote $\Phi=\Pr(y^{\prime}|\mathcal{S})$. By combining (\ref{equ:case1_upper}) and (\ref{equ:case2_upper}), the upper bound is
\begin{align*}
    \Phi_{\text{upper}}
    =\frac{1}{\ell}ap+(1-\frac{1}{\ell})\frac{b}{2}
\end{align*}
According to (\ref{equ:case1_lower}) and (\ref{equ:case2_lower}), the lower bound can be discussed by the following two cases. 
    
\textbf{Case 1.} If $a(1-p)<\frac{b}{2}$, i.e., $e^{\epsilon_1}<\frac{e^{\epsilon_2}+1}{2}$, we have
\begin{align*}
    \Phi_{\text{lower}} &= \frac{\eta}{|\mathcal{S}|}a(1-p)+(1-\frac{\eta}{|\mathcal{S}|})\frac{b}{2}\bigg|_{\frac{\eta}{|\mathcal{S}|}=\frac{1}{\ell}}\\
    &= \frac{1}{\ell}a(1-p)+(1-\frac{1}{\ell})\frac{b}{2}
\end{align*}
where $\Phi_{\text{lower}}<\frac{b}{2}$. Then, for any two different inputs $\mathcal{S}_1$ and $\mathcal{S}_2$, the probability ratio is bounded by
\begin{align}
\label{equ:bound1}
    \frac{\Pr(y^{\prime}|\mathcal{S}_1)}{\Pr(y^{\prime}|\mathcal{S}_2)}
    &\leqslant\frac{\Phi_{\text{upper}}}{\Phi_{\text{lower}}}
    =\frac{\frac{1}{\ell}ap+(1-\frac{1}{\ell})\frac{b}{2}}{\frac{1}{\ell}a(1-p)+(1-\frac{1}{\ell})\frac{b}{2}}\notag\\
    &=\frac{\frac{ap}{b}+\frac{\ell-1}{2}}{\frac{a(1-p)}{b}+\frac{\ell-1}{2}}
    =\frac{e^{\epsilon_1+\epsilon_2}+(\ell-1)\frac{e^{\epsilon_2}+1}{2}}{e^{\epsilon_1}+(\ell-1)\frac{e^{\epsilon_2}+1}{2}}
\end{align}

\textbf{Case 2.} If $a(1-p)\geqslant\frac{b}{2}$, i.e., $e^{\epsilon_1}\geqslant\frac{e^{\epsilon_2}+1}{2}$, then $\Phi_\text{lower}=\frac{b}{2}$
\begin{align}
\label{equ:bound2}
    &\quad\frac{\Pr(y^{\prime}|\mathcal{S}_1)}{\Pr(y^{\prime}|\mathcal{S}_2)}
    \leqslant\frac{\Phi_{\text{upper}}}{\Phi_{\text{lower}}}
    =\frac{\frac{1}{\ell}ap+(1-\frac{1}{\ell})\frac{b}{2}}{\frac{b}{2}}\notag\\
    &=\frac{2e^{\epsilon_1+\epsilon_2}}{\ell(e^{\epsilon_2}+1)}+1-\frac{1}{\ell}
    =\frac{e^{\epsilon_1+\epsilon_2}+(\ell-1)\frac{e^{\epsilon_2}+1}{2}}{\ell\cdot\frac{e^{\epsilon_2}+1}{2}}
\end{align}

By combining the results in (\ref{equ:bound1}) and (\ref{equ:bound2}), we have
\begin{align*}
    \frac{\Pr(y^{\prime}|\mathcal{S}_1)}{\Pr(y^{\prime}|\mathcal{S}_2)}\leqslant\frac{e^{\epsilon_1+\epsilon_2}+(\ell-1)(e^{\epsilon_2}+1)/2}{\min\{e^{\epsilon_1},(e^{\epsilon_2}+1)/2\}+(\ell-1)(e^{\epsilon_2}+1)/2}
\end{align*}
\end{proof}

\section{Proof of Theorem \ref{thm:estimation}}
\label{apx:thm:estimation_UE}
\begin{proof}
\textbf{Step 1. calculate the expectation and variance of $n_1$ and $n_2$.} 
Denote
\begin{align*}
    q_1=a\cdot[1+(2p-1)m_k^{*}]/2,\quad
    q_2=a\cdot[1-(2p-1)m_k^{*}]/2
\end{align*}
where $m_k^{*}$ is the true mean of key $k$. For a user $u\in\mathcal{U}_k$ (the set of users who possess key $k\in\mathcal{K}$), denote the expected contribution of supporting $1$ and $-1$ as $q_{u1}^{*}$ and $q_{u2}^{*}$ respectively. According to the perturbation steps of PCKV-UE in Figure \ref{fig:perturbation} (note that PCKV-GRR has the similar perturbation), $q_{u1}^{*}$ and $q_{u2}^{*}$ are computed by
\begin{align*}
    q_{u1}^{*}=a\cdot[1+(2p-1)v^{*}_u]/2,\quad
    q_{u2}^{*}=a\cdot[1-(2p-1)v^{*}_u]/2
\end{align*}
where $\frac{\sum_{u\in\mathcal{U}_k}q_{u1}^{*}}{|\mathcal{U}_k|}=q_1$ and $\frac{\sum_{u\in\mathcal{U}_k}q_{u2}^{*}}{|\mathcal{U}_k|}=q_2$. Then the expected contribution of supporting $1$ of a group of users $\mathcal{U}_k$ is
\begin{align*}
    \mathbb{E}_{\mathcal{U}_k}[n_1]=\frac{1}{\ell}\sum\nolimits_{u\in\mathcal{U}_k}q_{u1}^{*}
    =\frac{1}{\ell}|\mathcal{U}_k|q_1
    =n\frac{f_k^{*}}{\ell}q_1
\end{align*}
where $|\mathcal{U}_k|=nf_k^{*}$. And the corresponding variance is 
\begin{align*}
    &\text{Var}_{\mathcal{U}_k}[n_1]=\frac{1}{\ell}\sum_{u\in\mathcal{U}_k}q_{u1}^{*}(1-q_{u1}^{*})
    =\frac{1}{\ell}\left[\sum_{u\in\mathcal{U}_k}q_{u1}^{*}-\sum_{u\in\mathcal{U}_k}{q_{u1}^{*}}^2\right]\\
    &\leqslant\frac{1}{\ell}\left[\sum_{u\in\mathcal{U}_k}q_{u1}^{*}-\frac{1}{|\mathcal{U}_k|}(\sum_{u\in\mathcal{U}_k}{q_{u1}^{*}})^2\right]
    =n\frac{f_k^{*}}{\ell}q_1(1-q_1)
\end{align*}
where $\sum_{u\in\mathcal{U}_k}{q_{u1}^{*}}^2\geqslant\frac{1}{|\mathcal{U}_k|}(\sum_{u\in\mathcal{U}_k}{q_{u1}^{*}})^2$ from Cauchy-Schwarz inequality. Similarly, we can compute $\mathbb{E}_{\mathcal{U}_k}[n_2]$ and the upper bound of $\text{Var}_{\mathcal{U}_k}[n_2]$. Then, for all users, the expectation and the upper bound of variance are ($t=1$ or $2$)
\begin{align*}
    &\mathbb{E}[n_t]=\mathbb{E}_{\mathcal{U}_k}[n_t]+\mathbb{E}_{\mathcal{U}\backslash\mathcal{U}_k}[n_t]
    =n\frac{f_k^{*}}{\ell}q_t+n(1-\frac{f_k^{*}}{\ell})\frac{b}{2}\\
    &\text{Var}[n_t]\leqslant n\frac{f_k^{*}}{\ell}q_t(1-q_t)+n(1-\frac{f_k^{*}}{\ell})\frac{b}{2}(1-\frac{b}{2})
\end{align*}
where $\mathcal{U}\backslash\mathcal{U}_k$ denotes the set of users not in $\mathcal{U}_k$. Note that 
\begin{align*}
    &\text{Var}_{\mathcal{U}_k}[n_1]-\text{Var}_{\mathcal{U}_k}[n_2]
    =\frac{1}{\ell}\sum_{u\in\mathcal{U}_k}(q_{u1}^{*}-q_{u2}^{*})(1-q_{u1}^{*}-q_{u2}^{*})\\
    &=n\frac{f_k^{*}}{\ell}(q_1-q_2)(1-a)
    =n\frac{f_k^{*}}{\ell}(1-a)a(2p-1)m_k^{*}
\end{align*}
because of $q_{u1}^{*}+q_{u2}^{*}=a$ and $\sum_{u}(q_{u1}^{*}-q_{u2}^{*})=nf_k^{*}(q_1-q_2)$, where $q_1-q_2=a(2p-1)m_k^{*}$. Then, for all users $u\in\mathcal{U}$,
\begin{align*}
    &\text{Var}[n_1]-\text{Var}[n_2]=n\frac{f_k^{*}}{\ell}(1-a)a(2p-1)m_k^{*}\\
    &\text{Var}[n_1+n_2]=n\frac{f_k^{*}}{\ell}a(1-a)+n(1-\frac{f_k^{*}}{\ell})b(1-b)
\end{align*}
Note that $n_1$ and $n_2$ are correlated variables.

\textbf{Step 2. calculate the expectation and variance of frequency estimation.}  According to the frequency estimator in (\ref{equ:hat_f}), we have
\begin{align*}
    \mathbb{E}[\hat{f}_k]&=\frac{\mathbb{E}[n_1+n_2]/n-b}{a-b}\ell=\frac{\frac{f_k^{*}}{\ell}a+(1-\frac{f_k^{*}}{\ell})b-b}{a-b}\ell=f_k^{*}\\
    \text{Var}[\hat{f}_k]&=\frac{\ell^2 \text{Var}[n_1+n_2]}{n^2(a-b)^2}
    = \frac{\ell^2 b(1-b)}{n(a-b)^2} + \frac{\ell\cdot f_k^{*}(1-a-b)}{n(a-b)}
\end{align*}
which are equivalent to the results for itemset data in \cite{wang2018locally} (note that \cite{wang2018locally} focuses on the count $c_k=nf_k^{*}$ while we consider the proportion $f_k^{*}$). 

\textbf{Step 3. calculate the expectation and variance of mean estimation.} 
From the multivariate Taylor Expansions of functions of random variables \cite{casella2002statistical}, the expectation of quotient of two random variables $X$ and $Y$ can be approximated by
\begin{align}
    \label{equ:E[X/Y]}
    \mathbb{E}\left[\frac{X}{Y}\right]&\approx \frac{\mathbb{E}[X]}{\mathbb{E}[Y]}-\frac{\text{Cov}_{X,Y}}{\mathbb{E}[Y]^2}+ \frac{\mathbb{E}[X]}{\mathbb{E}[Y]^3}\cdot\text{Var}[Y]\\
    \label{equ:Var[X/Y]}
    \text{Var}\left[\frac{X}{Y}\right]&\approx \frac{\text{Var}[X]}{\mathbb{E}[Y]^2}-\frac{2\mathbb{E}[X]\text{Cov}_{X,Y}}{\mathbb{E}[Y]^3}+ \frac{\mathbb{E}[X]^2}{\mathbb{E}[Y]^4}\text{Var}[Y]
\end{align}
For convenience, denote $X=n_1-n_2,Y=n_1+n_2-nb$, then
\begin{align*}
    \mathbb{E}[X]=n\frac{f_k^{*}}{\ell}a(2p-1)m_k^{*},\quad
    \mathbb{E}[Y]=n\frac{f_k^{*}}{\ell}(a-b)
\end{align*}
The variances are
\begin{align*}
    \text{Var}[X]&=\text{Var}[n_1-n_2]=2(\text{Var}[n_1]+\text{Var}[n_2])-\text{Var}[n_1+n_2]\\
    &\leqslant nb+n\frac{f_k^{*}}{\ell}[(a-b)-a^2(2p-1)^2{m_k^{*}}^2]\\
    \text{Var}[Y]&=\text{Var}[n_1+n_2]
    =n\frac{f_k^{*}}{\ell}a(1-a)+n(1-\frac{f_k^{*}}{\ell})b(1-b)
\end{align*}
The covariance is 
\begin{align*}
   &\quad\text{Cov}_{X,Y}=\text{Cov}[n_1-n_2,n_1+n_2]\\
   &=\mathbb{E}[(n_1-n_2)(n_1+n_2)]-\mathbb{E}[n_1-n_2]\mathbb{E}[n_1+n_2]\\
   &=\mathbb{E}[n_1^2-n_2^2]-\left(\mathbb{E}[n_1]^2-\mathbb{E}[n_2]^2\right)=\text{Var}[n_1]-\text{Var}[n_2]\\
   &=n\frac{f_k^{*}}{\ell}a(1-a)(2p-1)m_k^{*}=(1-a)\cdot\mathbb{E}[X]
\end{align*}
Note that only $\text{Var}[X]$ is computed by its upper bound, while $\mathbb{E}[X]$, $\mathbb{E}[Y]$, $\text{Var}[Y]$ and $\text{Cov}_{X,Y}$ are computed by their exact values. For convenience, denote $\delta=\frac{f_k^{*}}{\ell}(a-b)$ and $\gamma=\frac{f_k^{*}}{\ell}a(2p-1)$. According to (\ref{equ:hat_m}) and (\ref{equ:E[X/Y]}), we have
\begin{align*}
    \mathbb{E}[\hat{m}_k]&=\frac{(a-b)\mathbb{E}\left[\frac{X}{Y}\right]}{a(2p-1)}
    \approx\frac{(a-b)\mathbb{E}[X]}{a(2p-1)\mathbb{E}[Y]}\left[1-\frac{1-a}{\mathbb{E}[Y]}+\frac{\text{Var}[Y]}{\mathbb{E}[Y]^2}\right]\\
    &=m_k^{*}\left[1+\frac{ (1-b-\delta)b }{n\delta^2}\right]
\end{align*}
Similarly, according to (\ref{equ:hat_m}) and (\ref{equ:Var[X/Y]}), we have
\begin{align*}
    \text{Var}[\hat{m}_k]=\frac{(a-b)^2\text{Var}\left[\frac{X}{Y}\right]}{a^2(2p-1)^2}
    \lesssim\frac{b+\delta}{n\gamma^2} + \frac{b(1-b)-\delta}{n\delta^2} \cdot{m_k^{*}}^2
\end{align*}
\end{proof}

\section{Proof of Lemma \ref{lem:hat_n}}
\label{apx:lem:hat_n}
\begin{proof}
According to the perturbation mechanism, we have
\begin{align*}
    \mathbb{E}[n_1]&=n_1^{*}ap+n_2^{*}a(1-p)+(n-n_1^{*}-n_2^{*})b/2\\
    \mathbb{E}[n_2]&=n_1^{*}a(1-p)+n_2^{*}ap+(n-n_1^{*}-n_2^{*})b/2
\end{align*}
which can be rewritten as
\begin{align*}
    \begin{bmatrix}
    \mathbb{E}[n_1]\\
    \mathbb{E}[n_2]
    \end{bmatrix}=
    A
    \begin{bmatrix}
    n_1^{*}\\
    n_2^{*}
    \end{bmatrix}+
    \begin{bmatrix}
    nb/2\\
    nb/2
    \end{bmatrix}
\end{align*}
where 
\begin{align*}
    A=
    \begin{bmatrix}
    ap-\frac{b}{2} & a(1-p)-\frac{b}{2}\\
    a(1-p)-\frac{b}{2} & ap-\frac{b}{2}
    \end{bmatrix}
\end{align*}
According to the linear property, the expectation of $\hat{n}_1$ and $\hat{n}_2$ in (\ref{equ:hat_n}) are
\begin{align*}
    \begin{bmatrix}
    \mathbb{E}[\hat{n}_1]\\
    \mathbb{E}[\hat{n}_2]
    \end{bmatrix}=
    A^{-1}
    \begin{bmatrix}
    \mathbb{E}[n_1]-nb/2\\
    \mathbb{E}[n_2]-nb/2
    \end{bmatrix}
    =A^{-1}A\begin{bmatrix}
    n_1^{*}\\
    n_2^{*}
    \end{bmatrix}=
    \begin{bmatrix}
    n_1^{*}\\
    n_2^{*}
    \end{bmatrix}
\end{align*}
Note that 
\begin{align*}
    \det(A)&=(ap-b/2)^2-(a(1-p)-b/2)^2\\
    &=a(a-b)(2p-1)>0
\end{align*}
thus $A^{-1}$ exists. Therefore, $(\hat{n}_1,\hat{n}_2)$ are unbiased estimators of $(n_1^{*},n_2^{*})$. 
\end{proof}

\section{Proof of Lemma \ref{lem:budget_UE}}
\label{apx:lem:budget_UE}
\begin{proof}
According to budget allocation in (\ref{equ:epsilon_1&2}) and perturbation probabilities setting of OUE, we can rewrite $a,b,p$ with respect to $\theta$
\begin{align*}
    a=\frac{1}{2},\quad
    b=\frac{1}{e^{\epsilon_1}+1}=\frac{1}{\theta+1},\quad
    p = \frac{1}{1+{e^{-\epsilon_2}}}
    =\frac{e^{\epsilon}}{2\theta}
\end{align*}
where $\frac{e^{\epsilon}+1}{2}\leqslant\theta <e^{\epsilon}$. Then, $g$ and $h$ in (\ref{equ:gh}) can be rewritten as the function of $\theta$
\begin{align*}
    g(\theta)=\frac{4}{(\theta+1)(e^{\epsilon}/\theta-1)^2},\quad
    h(\theta) =\frac{4\theta}{(\theta-1)^2}
\end{align*}
and their derivative functions are
\begin{align*}
    g^{\prime}(\theta)=\frac{4\theta[\theta^2+(\theta+2)e^{\epsilon}]}{(\theta+1)^2(e^{\epsilon}-\theta)^3}>0,\quad
    h^{\prime}(\theta)=-\frac{4(\theta+1)}{(\theta-1)^3}<0
\end{align*}
For convenience, denote
\begin{align}
    \label{equ:Phi_theta}
    \Phi(\theta)=\text{MSE}_{\hat{m}_k}/\mu=g(\theta)+h(\theta)\cdot{m_k^{*}}^2
\end{align}
and $\theta_0=\frac{e^{\epsilon}+1}{2}$, which is the minimum value of $\theta$. In the following, we show that $\Phi(\epsilon_1)$ is an approximately increasing function of $\epsilon_1$.

Considering both $g^{\prime}(\theta)$ and $h^{\prime}(\theta)$ are increasing functions of $\theta$, we have
\begin{align*}
    &\quad\Phi^{\prime}(\theta)=g^{\prime}(\theta)+h^{\prime}(\theta)\cdot{m_k^{*}}^2
    \geqslant g^{\prime}(\theta_0)+h^{\prime}(\theta_0)\cdot{m_k^{*}}^2\\
    &= \frac{16(e^\epsilon+3)}{(e^\epsilon-1)^3}\cdot\left[\frac{(e^\epsilon+1)(3e^{2\epsilon}+12e^\epsilon+1)}{(e^\epsilon+3)^3}-{m_k^{*}}^2\right]
\end{align*}
where $-1\leqslant m_k^{*}\leqslant 1$. Denote
\begin{align}
    \label{equ:Psi_epsilon}
    \Psi(\epsilon) = \frac{(e^\epsilon+1)(3e^{2\epsilon}+12e^\epsilon+1)}{(e^\epsilon+3)^3}
\end{align}
whose value is plotted in Figure \ref{fig:proof} (a),
where $0.5<\Psi(\epsilon)<3$ for all $\epsilon>0$, and $\Psi(0.85)\approx 1$. Therefore, we have $\quad\Phi^{\prime}(\theta)\geqslant0$ for all $\epsilon_1\in[\ln\frac{e^\epsilon+1}{2},\epsilon)$ when $\Psi(\epsilon)\geqslant {m_k^{*}}^2$, which always holds if ${m_k^{*}}^2\leqslant 0.5$ or $\epsilon\geqslant 0.85$. Moreover, with different $\epsilon$, the value of $\Phi(\theta)$ in (\ref{equ:Phi_theta}) when ${m_k^{*}}^2=1$ (the worst case) is shown in Figure \ref{fig:proof} (b), which validates that  $\Phi(\theta)$ is approximately increasing function of $\theta$ for all possible $\epsilon$ and $m_k^{*}$. Therefore, $\theta_0=\frac{e^{\epsilon}+1}{2}$ is the optimal solution of minimizing $\text{MSE}[\hat{m}_k]=\mu\cdot\Phi(\theta)$. By substituting $\theta=\frac{e^{\epsilon}+1}{2}$ into (\ref{equ:epsilon_1&2}), we finally obtain the budgets as in \eqref{equ:budget_UE}.

\begin{figure}[!t]
    \centering
    \subfloat[The value of $\Psi(\epsilon)$ in (\ref{equ:Psi_epsilon}), where the positive $\Psi(\epsilon)-{m_k^{*}}^2$ indicates that $\text{MSE}_{\hat{m}_k}$ is a monotonically increasing function of $\theta$.]{\includegraphics[width=1.6in]{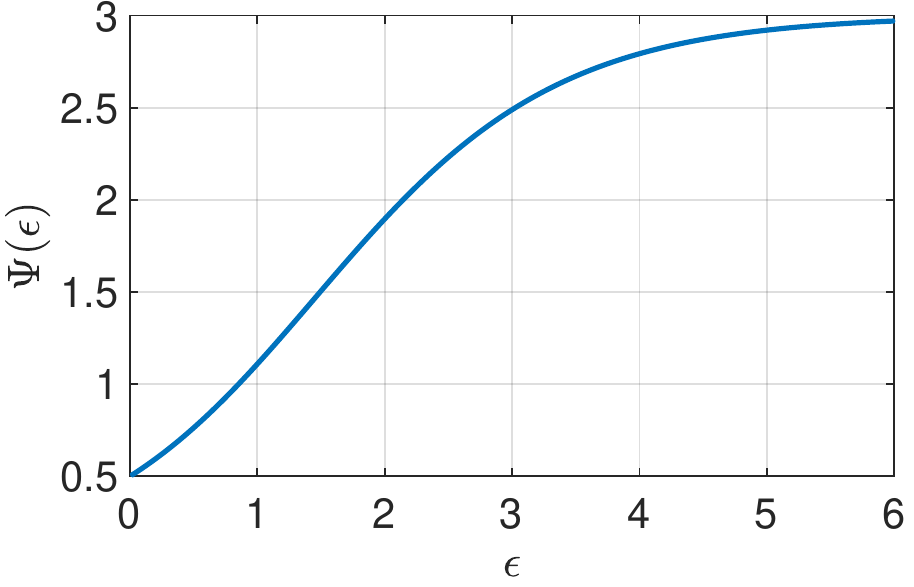}}
    \hfill
    \subfloat[The value of $\Phi(\theta)$ in (\ref{equ:Phi_theta}) when ${m_k^{*}}^2=1$ (the worst case), where $\Phi(\theta)=\text{MSE}_{\hat{m}_k}/\mu$ ($\mu$ is a constant) and $\theta\in[\frac{e^\epsilon+1}{2},e^{\epsilon})$ according to (\ref{equ:epsilon_1&2}).]{\includegraphics[width=1.6in]{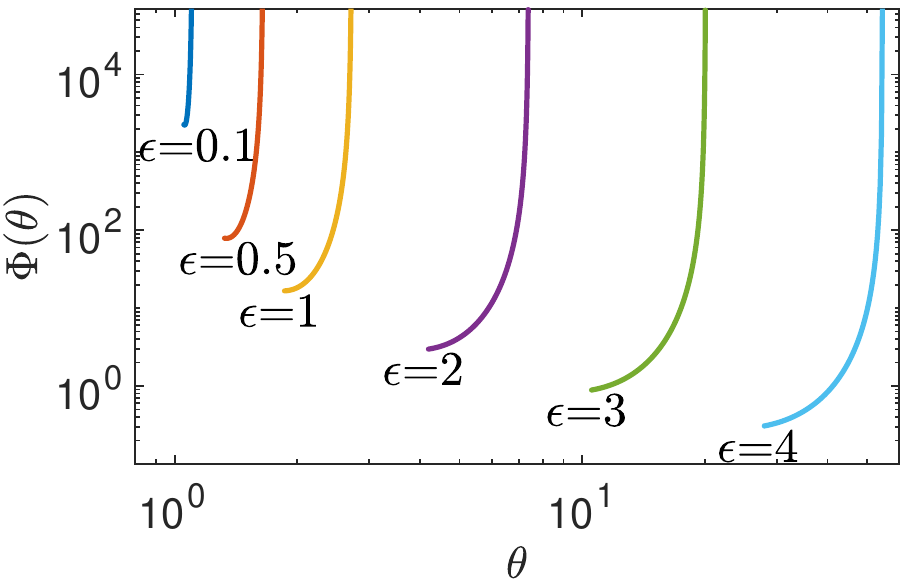}}\\
    \vspace{-3mm}
    \caption{Illustrations in Appendix \ref{apx:lem:budget_UE}.}
    \label{fig:proof}
\end{figure}

\end{proof}

\section{Proof of Lemma \ref{lem:budget_GRR}}
\label{apx:lem:budget_GRR}
\begin{proof}
According to (\ref{equ:abp_GRR}) and \eqref{equ:gh}, we have
\begin{align}
    \label{equ:proof_g}
    g&=\frac{b}{a^2(2p-1)^2}=
    \frac{e^{-\epsilon_1}+(d^{\prime}-1)e^{-2\epsilon_1}}{(\frac{2}{1+e^{-\epsilon_2}}-1)^2}\\
    \label{equ:proof_h}
    h&= \frac{(1-b)b}{(a-b)^2}=
    \frac{e^{\epsilon_1}+d^{\prime}-2}{(e^{\epsilon_1}-1)^2}
\end{align}
where $d^{\prime}=d+\ell$ and $h^{\prime}(\epsilon_1)<0$. In the following, we discuss the optimal $\epsilon_1$ in two cases.

\textbf{Case 1.} 
If $\ell=1$, then (\ref{equ:epsilon_GRR}) reduces to (\ref{equ:epsilon_UE}) because of $\lambda=0$, thus we can obtain the same result as in PCKV-UE
\begin{align}
    \epsilon_1=\ln \theta,\quad \epsilon_2=\ln\frac{1}{2\theta e^{-\epsilon}-1},\quad \text{for }\frac{e^{\epsilon}+1}{2}\leqslant \theta <e^{\epsilon}
\end{align}
then we have
\begin{align*}
    g(\theta)
    =\frac{\theta+(d^{\prime}-1)}{(e^{\epsilon}-\theta)^2},\quad
    h(\theta)=
    \frac{\theta+d^{\prime}-2}{(\theta-1)^2}
\end{align*}
where $g^{\prime}(\theta)>0$ and $h^{\prime}(\theta)<0$. Similar to the proof in Appendix \ref{apx:lem:budget_UE}, the optimal solution of minimizing $g(\theta)+h(\theta)\cdot{m_k^{*}}^2$ can be approximated at $\theta=\frac{e^{\epsilon}+1}{2}$, then $\epsilon_1=\ln\frac{e^{\epsilon}+1}{2}$ and $\epsilon_2=\epsilon$.

\textbf{Case 2.}
If $\ell>1$, denote $\theta=e^{\epsilon_1}$ and let $\epsilon_2=0$ in (\ref{equ:epsilon_GRR}), then
\begin{align*}
    e^{\epsilon}=\frac{\theta+(\ell-1)}{\ell}
    ~~\Rightarrow~~
    \theta=\ell\cdot(e^{\epsilon}-1)+1
\end{align*}
Thus, to guarantee $\epsilon_1,\epsilon_2>0$ under a given $\epsilon$,  variable $\theta$ should in the following range
\begin{align}
    \label{equ:range_theta}
    1<\theta<\ell\cdot(e^{\epsilon}-1)+1
\end{align}
On the other hand, let $\theta=e^{\epsilon_1}=(e^{\epsilon_2}+1)/2$ in (\ref{equ:epsilon_GRR}), i.e., the two values in the min operation equal with each other, then $\theta=\ell\cdot(e^\epsilon-1)/2+1$. For the parameter $g$ calculated in \eqref{equ:proof_g}, we discuss its derivative function $g^{\prime}(\theta)$ in the two ranges
\begin{itemize}
    \item For $1<\theta\leqslant\ell\cdot(e^\epsilon-1)/2+1$, we have 
    \begin{align*}
    \min\{e^{\epsilon_1},(e^{\epsilon_2}+1)/2\}=e^{\epsilon_1}
    ~~\Rightarrow~~
    e^\epsilon=\frac{e^{\epsilon_1+\epsilon_2}+\lambda}{e^{\epsilon_1}+\lambda}
    \end{align*}
    where $\lambda=(\ell-1)(e^{\epsilon_2}+1)/2$. Then,
    \begin{align*}
        &\frac{2}{1+e^{-\epsilon_2}}-1 
        = \frac{e^{\epsilon}-1}{e^{\epsilon}+1}\cdot[1+(\ell-1)/\theta]\\
        \Rightarrow~~
        &g(\theta)=\left(\frac{e^{\epsilon}+1}{e^{\epsilon}-1}\right)^2\cdot\frac{\theta+d^{\prime}-1}{(\theta+\ell-1)^2}
    \end{align*}
    where $g^{\prime}(\theta)<0$ and $d^{\prime}=d+\ell$.
    \item For $\ell\cdot(e^\epsilon-1)/2+1\leqslant\theta<\ell\cdot(e^\epsilon-1)+1$, we have
    \begin{align*}
    e^\epsilon=\frac{e^{\epsilon_1+\epsilon_2}+\lambda}{(e^{\epsilon_2}+1)/2+\lambda}
    \end{align*}
    then
    \begin{align*}
        &\frac{2}{1+e^{-\epsilon_2}}-1 
        = [\ell(e^{\epsilon}-1)-(\theta-1)]/\theta\\
        \Rightarrow~~
        &g(\theta)=\frac{\theta+d^{\prime}-1}{[\ell(e^{\epsilon}-1)-(\theta-1)]^2}
    \end{align*}
    where $g^{\prime}(\theta)>0$.
\end{itemize}
Therefore, $g(\theta)$ approaches to the minimum value at $\theta=\ell\cdot(e^\epsilon-1)/2+1$. Note that $g(\theta)\rightarrow+\infty$ when $\theta\rightarrow\ell\cdot(e^{\epsilon}-1)+1$ (the upper bound in \eqref{equ:range_theta}), and $h(\theta)\rightarrow+\infty$ when $\theta\rightarrow1$ (the lower bound in \eqref{equ:range_theta}). Similar to the proof in Appendix \ref{apx:lem:budget_UE}, the optimal solution of minimizing $g(\theta)+h(\theta)\cdot{m_k^{*}}^2$ can be approximated at $\theta=\ell\cdot(e^\epsilon-1)/2+1$. Then, we have
\begin{align*}
    \epsilon_1=\ln[\ell\cdot(e^\epsilon-1)/2+1],\quad
    \epsilon_2=\ln\left[\ell\cdot(e^\epsilon-1)+1\right]
\end{align*}

By combining the results in Case 1 (when $\ell=1$) and Case 2 (when $\ell>1$), we obtain (\ref{equ:budget_GRR}).
\end{proof}

\end{document}